\documentclass[submission,copyright,creativecommons]{eptcs}
\providecommand{\event}{QPL 2011} % Name of the event you are submitting to

\usepackage{amsthm}
\usepackage{amssymb}
\usepackage{hhline}
\usepackage[dvipdfm]{graphicx}

\newtheorem{dfn}{Definition}[section]
\newtheorem{thm}[dfn]{Theorem}
\newtheorem{lmm}[dfn]{Lemma}
\newtheorem{cor}[dfn]{Corollary}

\newtheorem{prp}[dfn]{Proposition}
\newtheorem{exm}[dfn]{Example}
\newcommand{\qbit}[1]{|#1\rangle}

\newcommand{\C}{\mathbb{C}}
\newcommand{\Hilbert}{\mathcal{H}}
\newcommand{\qin}[2]{\langle #1 | #2 \rangle}
\newcommand{\qeff}[1]{\langle#1|}
\newcommand{\qact}[2]{\qbit{#1}\qeff{#2}}

\newcommand{\mat}[9]{\left(
\begin{array}{ccc}
#1 & #2 & #3\\
#4 & #5 & #6\\
#7 & #8 & #9
\end{array}\right)}

\newcommand{\G}{\mathcal{G}}
\newcommand{\W}{\mathcal{W}}
\newcommand{\I}{\mathcal{I}}
\newcommand{\T}{\mathcal{T}}

\newcommand{\gph}[5]{\mbox{\raisebox{#1ex}[#2ex][#3ex]{\includegraphics[keepaspectratio=true, scale = #4]{image/#5.eps}}}}
\newcommand{\gphh}[5]{\mbox{\raisebox{#1ex}[#2ex][#3ex]{\includegraphics[keepaspectratio=true, height = #4ex]{image/#5.eps}}}}
\newcommand{\gmu}{\mbox{\raisebox{-1.2ex}[3ex][2ex]{\includegraphics[keepaspectratio=true, scale = 0.3]{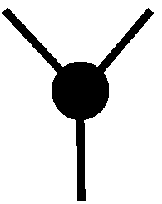}}}}
\newcommand{\gdelta}{\mbox{\raisebox{-1.2ex}[3ex][2ex]{\includegraphics[keepaspectratio=true, scale = 0.3]{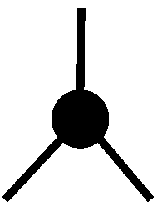}}}}
\newcommand{\geta}{\mbox{\raisebox{-1.0ex}[3ex][2ex]{\includegraphics[keepaspectratio=true, scale = 0.3]{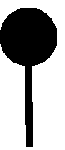}}}}
\newcommand{\gepsilon}{\mbox{\raisebox{-1.0ex}[3ex][2ex]{\includegraphics[keepaspectratio=true, scale = 0.3]{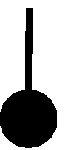}}}}

\title{Graphical Classification of Entangled Qutrits}
\author{Kentaro Honda
\institute{Dept.\ Computer Science,\ University of Tokyo, Japan}
\email{honda@is.s.u-tokyo.ac.jp}}

\def\titlerunning{Graphical Classification of Entangled Qutrits}
\def\authorrunning{K. Honda}
\begin{document}
\maketitle

\begin{abstract}
  A multipartite quantum state is entangled if it is not separable.
  Quantum entanglement plays a fundamental role in many applications
  of quantum information theory, such as quantum teleportation.
  Stochastic local quantum operations and classical communication
  (SLOCC) cannot essentially change quantum entanglement without
  destroying it. Therefore, entanglement can be classified by dividing
  quantum states into equivalence classes, where two states are
  equivalent if each can be converted into the other by SLOCC.
  Properties of this classification, especially in the case of non
  two-dimensional quantum systems, have not been well
  studied. Graphical representation is sometimes used to clarify the
  nature and structural features of entangled states. SLOCC
  equivalence of quantum bits (qubits) has been described graphically
  via a connection between tripartite entangled qubit states and
  commutative Frobenius algebras (CFAs) in monoidal categories. In
  this paper, we extend this method to qutrits, i.e., systems that
  have three basis states. We examine the correspondence between CFAs
  and tripartite entangled qutrits. Using the symmetry property, which
  is required by the definition of a CFA, we find that there are only
  three equivalence classes that correspond to CFAs. We represent
  qutrits graphically, using the connection to CFAs. We derive
  equations that characterize the three equivalence classes. Moreover,
  we show that any qutrit can be represented as a composite of three
  graphs that correspond to the three classes.
\end{abstract}

\section{Introduction}

In quantum computing, it may be necessary to send quantum information,
but quantum information cannot be replicated
\cite{Wootters1982}. Foundational methods of quantum information
theory, such as quantum teleportation \cite{PhysRevLett.70.1895},
provide a way to send quantum information using quantum entanglement,
i.e., states that cannot be separated. Quantum entanglement is a
non-local property of quantum states, so entanglement does not
increase by stochastic local quantum operations and classical
communication (SLOCC). The entangled states are divided into
equivalence classes by SLOCC-equivalence, which relates states that can be
converted into each other by SLOCC. Several recent studies have investigated SLOCC-equivalence classes \cite{PhysRevA.62.062314}\cite{PhysRevA.74.052336}\cite{PhysRevA.75.022318}\cite{PhysRevA.81.052315}\cite{PhysRevA.65.052112}\cite{qutrit}.
In this paper, we focus on tripartite qutrits. Qutrits are systems
that have three-dimensional state spaces. They are not as well-studied
as qubits
\cite{PhysRevLett.85.3313}\cite{PhysRevA.64.012306}\cite{PhysRevLett.88.127901}.
To describe and clarify the SLOCC-equivalence classes of qutrits and
their structural features, we express them graphically.

Morphisms of categories can be expressed graphically \cite{Selinger09Survey}. Moreover, recently, quantum
protocols and quantum computing have been interpreted in monoidal categories,
which are categories with a tensor product. Abramsky and Coecke gave
quantum axioms and an interpretation of quantum protocols in a kind of
monoidal category called a biproduct dagger compact closed category
\cite{paper72}. Selinger gave a categorical semantics for a quantum
programming language QPL \cite{Selinger04}\cite{Selinger2007139}.
Using graphical representations of morphisms in monoidal categories,
and connections between quantum information theory and category
theory, qutrits can also be expressed graphically.

For graphical expression, we adopted an extension of a previous method \cite{BCAK}, in which highly entangled and highly symmetric quantum systems correspond to commutative Frobenius algebras (CFAs) in monoidal categories. That study employed a graphical representation of qubits.
This representation reflects the degree of entangledness.
In this paper, we used qutrits.
Although the classes of tripartite qutrits are infinite, there are a limited number of SLOCC-equivalence classes of tripartite qubits.
Using symmetry, we found that only three classes corresponded to CFAs.
The classifications are based on the algebraic structure of some kinds of qutrits.
We characterized each of these by applying the equations and graphing the results; we obtained three graphs.
Finally, we showed that any qutrits can be expressed graphically using the three graphs and single qutrits.

The remainder of this paper is organized as follows. In section 2, we define SLOCC-equivalence and describe the infinite classes of tripartite qutrits.
In section 3, we provide graphical representations of a monoidal category and a CFA, and describe some theorems as well as the qubits used in a previous study \cite{BCAK}.
Our results are presented in section 4.
We classify SLOCC-equivalence classes into non-maximal, non-symmetric, Frobenius, or other classes, and show that there are only three Frobenius classes. We also define ISCFA, a type of CFA, and prove the uniqueness of the correspondence between Frobenius classes and the three CFAs (SCFA, ACFA, and ISCFA). Finally, we demonstrate how to construct a qutrit.

%-------------------
% Section 2.
%-------------------
\section{SLOCC-equivalence}
Whether or not a state is entangled is an important question. Another important question is how entangled the state is.
For example, we consider the following two tripartite qubits. In the following, we omit the normalization factor of states.
\begin{equation}
\qbit{\mbox{GHZ}}:= \qbit{000}+\qbit{111} \label{GHZ},
\end{equation}
\begin{equation}
\qbit{\mbox{W}}:= \qbit{001}+\qbit{010}+\qbit{100} \label{W}.
\end{equation}
There are two states: the GHZ state and the W state.
Let the first qubit of each tripartite qubit  be observed with respect
to the canonical basis of $\C^2$.
After the GHZ state is observed, the changed state is $\qbit{000}$ or $\qbit{111}$. Both states are separable.
However, after the W state is observed, the changed state is
$\qbit{01}+\qbit{10}$ or $\qbit{00}$. $\qbit{00}$ is separable, but
$\qbit{01}+\qbit{10}$ is entangled. Hence, both the GHZ state and the
W state are entangled, but to different degrees. We need to classify systems based on their degree of entanglement.

Entanglement is a property of a multi-partite system, so local operations within each system cannot essentially change entanglement without destroying it. 
Classical communication does not change the properties of a system. Therefore, quantum local operations and classical communication (LOCC) do not essentially change the entanglement of systems without destroying it.
If a state $\qbit{\psi}$ can be converted into $\qbit{\phi}$ by LOCC
with non-zero probability, we say that $\qbit{\psi}$ can be converted into $\qbit{\phi}$ by stochastic local quantum operations and classical communication (SLOCC).
Using SLOCC, we can define an equivalence relation on entangled states.

\begin{dfn}[SLOCC-equivalence]\rm
If states $\qbit{\psi}$ and $\qbit{\phi}$ can be converted into each other by SLOCC, then they are {\em SLOCC-equivalent}. 
\end{dfn}

Moreover, based on SLOCC-equivalence, SLOCC-maximality can be defined.

\begin{dfn}[SLOCC-maximality]\label{SLOCCmaximal}\rm
Let $\qbit{\psi}$ be a state.
If every $\qbit{\phi}$ that can be converted into $\qbit{\psi}$ by SLOCC is SLOCC-equivalent to $\qbit{\psi}$, then $\qbit{\psi}$ is SLOCC-maximal.
\end{dfn}

SLOCC-equivalence is an equivalence relation, and therefore it determines a notion of equivalence class (SLOCC class). When $x$ is a representative of a SLOCC class, we use $\overline{x}$ to indicate the SLOCC class.
It has been shown that $N$-partite systems $\qbit{\psi}$ and $\qbit{\phi}$ are SLOCC-equivalent iff there exist invertible matrices $L_1,\ldots,L_n$ such that $\qbit{\psi} = \bigotimes_{i=1}^N L_i \qbit{\phi}$ \cite{PhysRevA.62.062314}.

There are six possible SLOCC classes of tripartite qubits, namely $\overline{\qbit{000}}$, $\overline{\qbit{000}+\qbit{011}}$, $\overline{\qbit{000}+\qbit{101}}$, $\overline{\qbit{000}+\qbit{110}}$, $\overline{\qbit{\mbox{GHZ}}}$, and $\overline{\qbit{\mbox{W}}}$. In contrast, there is an infinite number of SLOCC classes of the tripartite qutrits in which we are interested here.
A previous study used an inductive method \cite{PhysRevA.74.052336} to identify some SLOCC classes of tripartite qutrits \cite{qutrit}.
These classes are shown in Table~\ref{SLOCCTable} below. In the table, $\phi$, $\varphi$, $\chi$, and $\psi$ are unit vectors of $\C^3$. Notice that $\overline{\qbit{\pi(\phi, \varphi, \chi, \psi)}}$ indicates an infinite number of SLOCC classes. 
$\overline{\qbit{\phi_{03}(\phi, \varphi)}}$, $\overline{\qbit{\varphi_2(\phi, \varphi)}}$, and $\overline{\qbit{\phi_{04}(\phi, \varphi)}}$ also indicate infinite families of classes.
\begin{table*}
\begin{center}
\begin{tabular}{||c|c|c||c|c||}
\hhline{|t:=:=:=:t:=:=:t|}
Name & \multicolumn{2}{c||}{Representative} & Name & Representative
\\\hhline{||-|-|-||-|-||}
\hhline{||-|-|-||-|-||}
\hhline{||-|-|-||-|-||}
\multicolumn{1}{||c|}{$\qbit{\psi_{0}}$} & \multicolumn{2}{c||}{$\qbit{000}$}
&
\multicolumn{1}{c|}{$\qbit{\psi_{12}}$} & \multicolumn{1}{c||}{$\qbit{000} + \qbit{011} + \qbit{101}+\qbit{112}$}
\\\hhline{||-|-|-||-|-||}
\multicolumn{1}{||c|}{$\qbit{\psi_{1}}$} & \multicolumn{2}{c||}{$\qbit{000} + \qbit{011}$}
&
\multicolumn{1}{c|}{$\qbit{\psi_{13}}$} & \multicolumn{1}{c||}{$\qbit{000} + \qbit{011} + \qbit{112}+\qbit{120}$}
\\\hhline{||-|-|-||-|-||}
\multicolumn{1}{||c|}{$\qbit{\psi_{2}}$} & \multicolumn{2}{c||}{$\qbit{000} + \qbit{011} + \qbit{022}$}
&
\multicolumn{1}{c|}{$\qbit{\psi_{14}}$} & \multicolumn{1}{c||}{$\qbit{000} + \qbit{011} + \qbit{120}+\qbit{101}$}
\\\hhline{||-|-|-||-|-||}
\multicolumn{1}{||c|}{$\qbit{\psi_{3}}$} & \multicolumn{2}{c||}{$\qbit{000} + \qbit{101}$}
&
\multicolumn{1}{c|}{$\qbit{\psi_{15}}$} & \multicolumn{1}{c||}{$\qbit{000} + \qbit{011} + \qbit{120}+\qbit{102}$}
\\\hhline{||-|-|-||-|-||}
\multicolumn{1}{||c|}{$\qbit{\psi_{4}}$} & \multicolumn{2}{c||}{$\qbit{000} + \qbit{110}$}
&
\multicolumn{1}{c|}{$\qbit{\psi_{16}}$} & \multicolumn{1}{c||}{$\qbit{000} + \qbit{011} + \qbit{022} + \qbit{101}$}
\\\hhline{||-|-|-||-|-||}
\multicolumn{1}{||c|}{$\qbit{\psi_{5}}$} & \multicolumn{2}{c||}{$\qbit{000} + \qbit{111}$}
&
\multicolumn{1}{c|}{$\qbit{\psi_{17}}$} & \multicolumn{1}{c||}{$\qbit{000} + \qbit{011} + \qbit{022} + \qbit{101} + \qbit{112}$}
\\\hhline{||-|-|-||-|-||}
\multicolumn{1}{||c|}{$\qbit{\psi_{6}}$} & \multicolumn{2}{c||}{$\qbit{000} + \qbit{011} + \qbit{101}$}
&
\multicolumn{1}{c|}{$\qbit{\psi_{18}}$} & \multicolumn{1}{c||}{$\qbit{000} + \qbit{011} + \qbit{022} + \qbit{112} + \qbit{120}$}
\\\hhline{||-|-|-||-|-||}
\multicolumn{1}{||c|}{$\qbit{\psi_{7}}$} & \multicolumn{2}{c||}{$\qbit{000} + \qbit{011} + \qbit{112}$}
&
\multicolumn{1}{c|}{$\qbit{\psi_{19}}$} & \multicolumn{1}{c||}{$\qbit{000} + \qbit{011} + \qbit{022} + \qbit{120} + \qbit{101}$}
\\\hhline{||-|-|-||-|-||}
\multicolumn{1}{||c|}{$\qbit{\psi_{8}}$} & \multicolumn{2}{c||}{$\qbit{000} + \qbit{011} + \qbit{120}$}
&
\multicolumn{1}{c|}{$\qbit{\psi_{20}}$} & \multicolumn{1}{c||}{$\qbit{000} + \qbit{011} + \qbit{122}$}
\\\hhline{||-|-|-||-|-||}
\multicolumn{1}{||c|}{$\qbit{\psi_{9}}$} & \multicolumn{2}{c||}{$\qbit{000} + \qbit{101} + \qbit{202}$}
&
\multicolumn{1}{c|}{$\qbit{\psi_{21}}$} & \multicolumn{1}{c||}{$\qbit{000} + \qbit{110} + \qbit{220}$}
\\\hhline{||-|-|-||-|-||}
\multicolumn{1}{||c|}{$\qbit{\psi_{10}}$} & \multicolumn{2}{c||}{$\qbit{000} + \qbit{111} + \qbit{202}$}
&
\multicolumn{1}{c|}{$\qbit{\psi_{22}}$} & \multicolumn{1}{c||}{$\qbit{000} + \qbit{111} + \qbit{220}$}
\\\hhline{||-|-|-||-|-||}
\multicolumn{1}{||c|}{$\qbit{\psi_{11}}$} & \multicolumn{2}{c||}{$\qbit{000} + \qbit{111} + \qbit{201}$}
&
\multicolumn{1}{c|}{$\qbit{\G}$} & \multicolumn{1}{c||}{$\qbit{000} + \qbit{111} + \qbit{222}$}
\\\hhline{|:=:=:=:b:=:=:|}
\multicolumn{2}{||c|}{Name} & 
\multicolumn{3}{c||}{Representative}
\\\hhline{||-|-|-|-|-||}
\hhline{||-|-|-|-|-||}
\hhline{||-|-|-|-|-||}
\multicolumn{2}{||c|}{$\qbit{\pi(\phi, \varphi, \chi, \psi)}$} & \multicolumn{3}{c||}{$\qbit{000} + \qbit{011} + \qbit{1\phi\varphi} + \qbit{2\chi\psi}$}
\\\hhline{||-|-|-|-|-||}
\multicolumn{2}{||c|}{$\qbit{\phi_{0}}$} & \multicolumn{3}{c||}{$\qbit{000} + \qbit{011} + \qbit{022} + \qbit{101} + \qbit{202}$}
\\\hhline{||-|-|-|-|-||}
\multicolumn{2}{||c|}{$\qbit{\phi_{1}}$} & \multicolumn{3}{c||}{$\qbit{000} + \qbit{011} + \qbit{022} + \qbit{110} + \qbit{220}$}
\\\hhline{||-|-|-|-|-||}
\multicolumn{2}{||c|}{$\qbit{\varphi_{1}}$} & \multicolumn{3}{c||}{$\qbit{000} + \qbit{011} + \qbit{022} + \qbit{101} + \qbit{212}$}
\\\hhline{||-|-|-|-|-||}
\multicolumn{2}{||c|}{$\qbit{\phi_{2}(\phi, \varphi)}$} & \multicolumn{3}{c||}{$\qbit{000} + \qbit{011} + \qbit{101} + \qbit{112} + \qbit{2\phi\varphi}$}
\\\hhline{||-|-|-|-|-||}
\multicolumn{2}{||c|}{$\qbit{\varphi_{2}(\phi, \varphi)}$} & \multicolumn{3}{c||}{$\qbit{000} + \qbit{011} + \qbit{112} + \qbit{120} + \qbit{2\phi\varphi}$}
\\\hhline{||-|-|-|-|-||}
\multicolumn{2}{||c|}{$\qbit{\phi_{3}(\phi, \varphi)}$} & \multicolumn{3}{c||}{$\qbit{000} + \qbit{011} + \qbit{120} + \qbit{101} + \qbit{2\phi\varphi}$}
\\\hhline{||-|-|-|-|-||}
\multicolumn{2}{||c|}{$\qbit{\phi_{4}}$} & \multicolumn{3}{c||}{$\qbit{000} + \qbit{011} + \qbit{101} + \qbit{112} + \qbit{202} + \qbit{221}$}
\\\hhline{||-|-|-|-|-||}
\multicolumn{2}{||c|}{$\qbit{\psi_{23}}$} & \multicolumn{3}{c||}{$\qbit{000} + \qbit{011} + \qbit{101} + \qbit{112} + \qbit{210} + \qbit{202}$}
\\\hhline{||-|-|-|-|-||}
\multicolumn{2}{||c|}{$\qbit{\phi_{5}}$} & \multicolumn{3}{c||}{$\qbit{000} + \qbit{011} + \qbit{101} + \qbit{112} + \qbit{221} + \qbit{210}$}
\\\hhline{||-|-|-|-|-||}
\multicolumn{2}{||c|}{$\qbit{s_0}$} & \multicolumn{3}{c||}{$\qbit{000} + \qbit{011} + \qbit{112} + \qbit{120} + \qbit{202} + \qbit{221}$}
\\\hhline{||-|-|-|-|-||}
\multicolumn{2}{||c|}{$\qbit{\phi_{6}}$} & \multicolumn{3}{c||}{$\qbit{000} + \qbit{011} + \qbit{112} + \qbit{120} + \qbit{221} + \qbit{210}$}
\\\hhline{||-|-|-|-|-||}
\multicolumn{2}{||c|}{$\qbit{\psi_{24}}$} & \multicolumn{3}{c||}{$\qbit{000} + \qbit{011} + \qbit{120} + \qbit{101} + \qbit{221} + \qbit{210}$}
\\\hhline{||-|-|-|-|-||}
\multicolumn{2}{||c|}{$\qbit{\phi_{7}}$} & \multicolumn{3}{c||}{$\qbit{000} + \qbit{011} + \qbit{022} + \qbit{101} + \qbit{112} + \qbit{202} + \qbit{221}$}
\\\hhline{||-|-|-|-|-||}
\multicolumn{2}{||c|}{$\qbit{\phi_{8}}$} & \multicolumn{3}{c||}{$\qbit{000} + \qbit{011} + \qbit{022} + \qbit{101} + \qbit{112} + \qbit{210} + \qbit{202}$}
\\\hhline{||-|-|-|-|-||}
\multicolumn{2}{||c|}{$\qbit{s_1}$} & \multicolumn{3}{c||}{$\qbit{000} + \qbit{011} + \qbit{022} + \qbit{101} + \qbit{112} + \qbit{221} + \qbit{210}$}
\\\hhline{||-|-|-|-|-||}
\multicolumn{2}{||c|}{$\qbit{w_0}$} & \multicolumn{3}{c||}{$\qbit{000} + \qbit{011} + \qbit{022} + \qbit{101} + \qbit{112} + \qbit{202}$}
\\\hhline{||-|-|-|-|-||}
\multicolumn{2}{||c|}{$\qbit{\varphi_{3}}$} & \multicolumn{3}{c||}{$\qbit{000} + \qbit{011} + \qbit{022} + \qbit{101} + \qbit{112} + \qbit{220}$}
\\\hhline{||-|-|-|-|-||}
\multicolumn{2}{||c|}{$\qbit{\phi_{9}}$} & \multicolumn{3}{c||}{$\qbit{000} + \qbit{011} + \qbit{022} + \qbit{101} + \qbit{112} + \qbit{221}$}
\\\hhline{|b:=:=:=:=:=:b|}
\end{tabular}
\caption{Representatives in SLOCC classes of tripartite qutrits}\label{SLOCCTable}
\end{center}
\end{table*}

%-------------------
% Section 3.
%-------------------
\section{Commutative Frobenius Algebra}

In this section, we present the mathematical basis for our study of qutrits, most of which is based on previous research \cite{BCAK}.
First, we provide a graphical representation of symmetric monoidal categories. Then we define a CFA in such a category, and present a graphical representation of this algebra.
Finally, we describe the special states that correspond to the algebra, and then classify the CFAs.

\subsection{Graphical representation of symmetric monoidal categories}

To represent quantum systems graphically, we use the graphical representation of symmetric monoidal categories.

\begin{dfn}[Symmetric Monoidal Category]\rm
A monoidal category $M$ consists of

(i) a category $C$;

(ii) a bifunctor tensor product $\otimes: C \times C \rightarrow C$;

(iii) a unit object $e \in C$;

(iv) a natural isomorphism $\lambda_a: e \otimes a \cong a$;

(v) a natural isomorphism $\rho_a: a \otimes e \cong a$;

(vi) a natural isomorphism $\alpha_{a,b, c}: a \otimes (b \otimes c) \cong (a \otimes b) \otimes c$.

\noindent
$M$ satisfies the following two equations.
\begin{eqnarray}
\alpha_{a \otimes b, c, d} \circ \alpha_{a,b, c \otimes d} &=& (\alpha_{a,b, c} \otimes 1_a) \circ \alpha_{a, b \otimes c, d} \circ (1_a\otimes\alpha_{b,c, d})\\
\lambda_b &=& \rho_a \circ \alpha_{a,e,b}
\end{eqnarray}
If $M$ has a natural isomorphism $\gamma_{a, b}: a \otimes b \cong b \otimes a$ such that
the following three equations hold, then $M$ is called a symmetric monoidal category:
\begin{eqnarray}
\lambda_a \circ \gamma_{a, e} &=& \rho_a\\
\gamma_{b, a} \circ \gamma_{a, b} &=& 1_{a \otimes b}\\
\alpha_{a, c, b} \circ (1_a \otimes \gamma_{b, c}) \circ \alpha^{-1}_{a, b, c} &=& (\gamma_{c, a} \otimes 1_{b}) \circ \alpha_{c, a, b} \circ \gamma_{a\otimes b, c}
\end{eqnarray}.
\end{dfn}

\begin{thm}[Coherence Theorem \cite{MacLane}]
The coherence theorem states that any diagram that is composed of the maps $\alpha$, $\lambda$, $\rho$, and $\gamma$, and their tensor products, commutes.
\end{thm}

This means that any two objects that are tensor products of $a_1,\ldots, a_n$ and $e$ can be identified, even if they differ in bracketing, or in the number of position of $e$'s.

\begin{exm}
$\mathbf{FdHilb}$ is a monoidal category whose objects are finite dimensional Hilbert spaces, and whose arrows are all linear functions, with the usual tensor product and the unit object $\C$.
\end{exm}

Arrows in monoidal categories can be expressed graphically \cite{Selinger09Survey}.
We assume a flow from top to bottom.
An object $a$ is written as a line.
\begin{equation}
\gph{-3.2}{3.2}{2.5}{0.34}{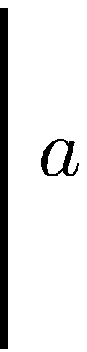}
\end{equation}
The unit object $e$ is expressed as no wire.
The tensor product $a \otimes b$ is written as two lines.
\begin{equation}
\gph{-4}{4}{4}{0.39}{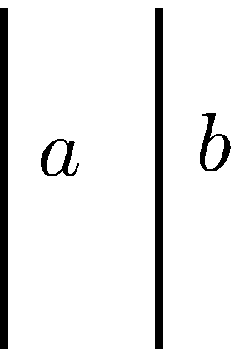}
\end{equation}

An arrow $f: a_1 \otimes \cdots \otimes a_n \rightarrow b_1 \otimes \cdots \otimes b_m$ is written as a box labeled $f$ from input wires $a_1,\ldots,a_n$ to output wires $b_1,\ldots,b_m$. 
\begin{equation}
\gph{-6.5}{5.8}{5}{0.39}{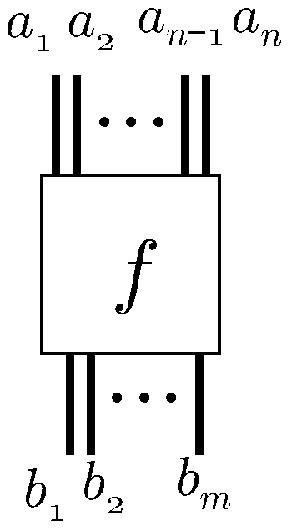}
\end{equation}
$e$ is not written graphically; hence, according to the coherence theorem, $\alpha_{a,b,c}$, $\lambda_a$, and $\rho_a$ are graphically expressed in the same way as identity arrows.
An identity arrow $1_a$ is written as a wire $a$.
\begin{equation}
\gph{-3.2}{3.2}{2.5}{0.34}{linea}
\end{equation}
The composition $g \circ f$ of arrows $f: a_1 \otimes \cdots \otimes a_n \rightarrow b_1 \otimes \cdots \otimes b_m$ and $g: b_1 \otimes \cdots \otimes b_m \rightarrow c_1 \otimes \cdots \otimes c_l$ is expressed as a vertical juxtaposition.
\begin{equation}
\gph{-10}{9}{8.5}{0.39}{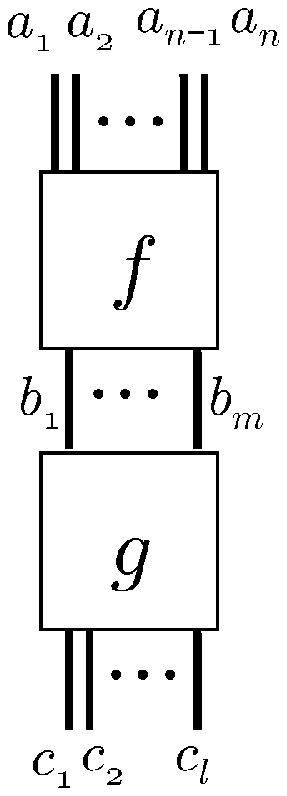}
\end{equation}
Two wires can be connected only if they are labeled by the same object.
The tensor product $f \otimes g: a \otimes b \rightarrow c \otimes d$ of arrows $f: a \rightarrow c$ and $g: b \rightarrow d$ is expressed as a horizontal juxtaposition.
\begin{equation}
\gph{-4.5}{5}{4}{0.39}{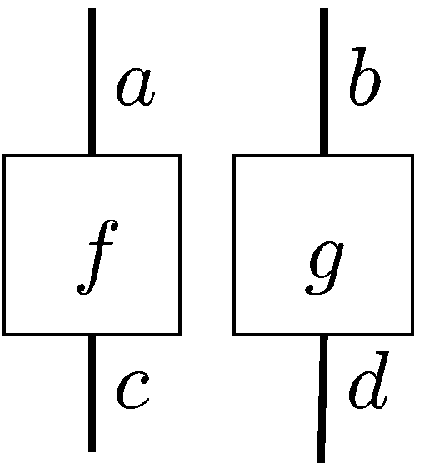}
\end{equation}
Furthermore, the natural isomorphism $\gamma_{a,b}$ of symmetric monoidal categories is expressed as an intersection.
\begin{equation}
\gph{-5}{7}{4}{0.34}{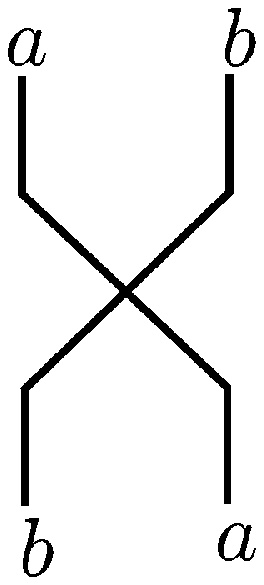}
\end{equation}
Using these expressions, all arrows of monoidal categories can be represented graphically.
Generally, a state vector $\qbit{\psi}$ of a state space $a$ can be considered a function from $\C$ to $a$.
Specifically, emphasizing no-input, $\qbit{\psi}$ is written as a triangle.
\begin{equation}
\gph{-3}{3.5}{2}{0.4}{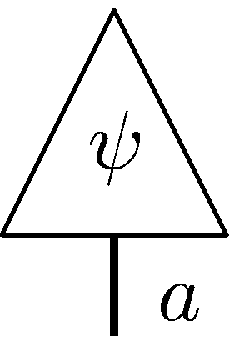}
\end{equation}
$\qeff{\psi}$ is written as a reversed triangle.
\begin{equation}
\gph{-3}{3.5}{2}{0.4}{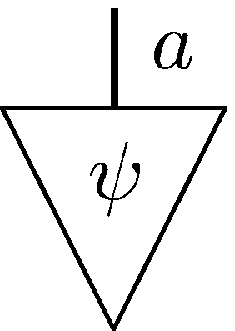}
\end{equation}

\subsection{Commutative Frobenius Algebra}
Using the graphical representation of monoidal categories, all arrows of monoidal categories can be represented graphically.
Therefore, if systems correspond to arrows of monoidal categories, then any system can be represented graphically.
Previous research \cite{BCAK} suggests that some kinds of system strictly correspond to a specific kind of algebra.

\begin{dfn}[Commutative Frobenius Algebra]\rm
A {\em Frobenius algebra} $F$ in a monoidal category $M$ consists of: 
\begin{enumerate}
\item[(i)] an object $a \in M$;
\item[(ii)] multiplication $\mu: a \otimes a \rightarrow a$;
\item[(iii)] a unit $\eta: e \rightarrow a$;
\item[(iv)] a comultiplication $\delta: a \rightarrow a \otimes a$;
\item[(v)] a counit: $\epsilon: a \rightarrow e$, where $e$ is the unit object of $M$;
\end{enumerate}
such that $F$ makes the following equations hold:
\begin{eqnarray}
(\delta \otimes 1_a) \circ \delta &=& \alpha_{a,a,a} \circ (1_a \otimes \delta) \circ \delta\\
\lambda_a \circ (\epsilon \otimes 1_a) \circ \delta &=& 1_a = \rho_a \circ (1_a \otimes \epsilon) \circ \delta\\
\mu \circ (\mu \otimes 1_a) \circ \alpha_{a,a,a} &=& \mu \circ (1_a \otimes \mu)\\
\mu \circ (\eta \otimes 1_a) \circ \lambda^{-1}_a &=& 1_a = \mu \circ (1_a \otimes \eta) \circ \rho^{-1}_a\\
(1_a \otimes \mu) \circ \alpha^{-1}_{a,a,a} \circ (\delta \otimes 1_a)
&=&
\delta \circ \mu
=
(\mu \otimes 1_a) \circ \alpha_{a,a,a} \circ (1_a \otimes \delta)
\end{eqnarray}
If $M$ is a symmetric monoidal category and $F$ satisfies the following equations, then $F$ is called a {\em commutative Frobenius algebra (CFA)}: 
\begin{eqnarray}
\mu \circ \gamma_{a,a} &=& \mu\\
\gamma_{a,a} \circ \delta &=& \delta
\end{eqnarray}
\end{dfn}

\begin{dfn}[$F$-graph]\rm
An $F$-graph of a CFA $F$ is an arrow that is composed of $\mu$, $\delta$, $\epsilon$, $\delta$, $\alpha_{a,a,a}$, $\rho_a$, $\lambda_a$, $\gamma_{a,a}$, $1_a$, and their tensor products.
\end{dfn}
The domain-codomain pairs of all components of a CFA differ from each other. Therefore, without labeling each arrow, we can present a CFA graphically, as follows.
\begin{eqnarray}
\begin{array}{cccc}
\mu = \gmu & \eta = \geta & \delta = \gdelta & \epsilon = \gepsilon
\end{array}
\end{eqnarray}

Using the representation of CFAs and monoidal categories, any $F$-graph can be represented graphically. 
Of course, the axioms of CFAs can be expressed graphically as follows. 
\begin{equation}
\begin{array}{l}
\gph{-2}{4}{2}{0.3}{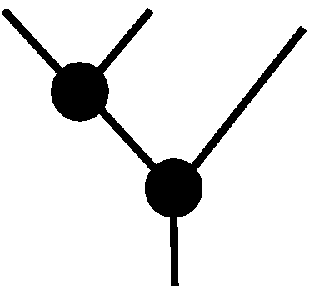} = \gph{-2}{4}{2}{0.3}{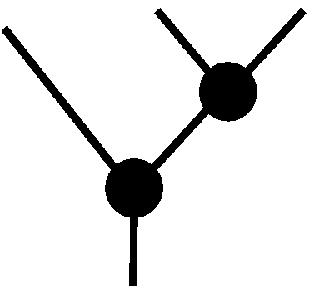}
\hspace{20pt}
\gph{-2}{4}{2}{0.3}{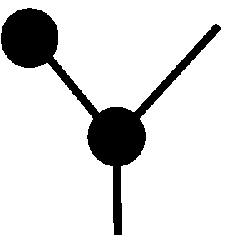} = \gph{-3}{4}{2}{0.3}{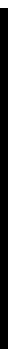} = \gph{-2}{4}{2}{0.3}{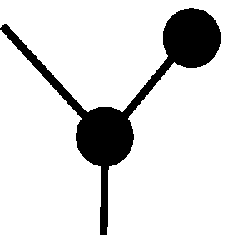}
\hspace{20pt}
\gph{-2}{4}{2}{0.3}{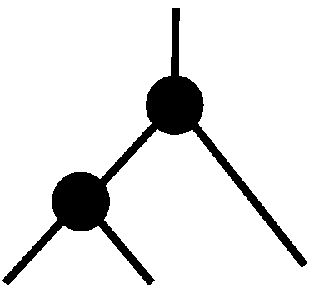} = \gph{-2}{4}{2}{0.3}{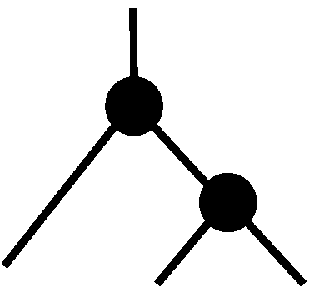}
\hspace{20pt}
\gph{-2}{4}{2}{0.3}{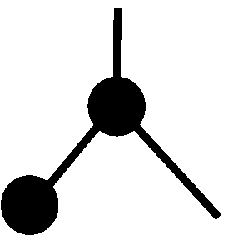} = \gph{-3}{4}{2}{0.3}{line} = \gph{-2}{4}{2}{0.3}{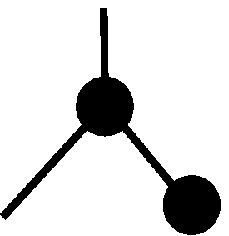}
\\
\\
\hspace{60pt}\gph{-2}{4}{2}{0.3}{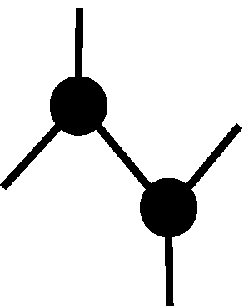} = \gph{-2}{4}{2}{0.3}{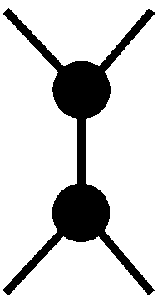} = \gph{-2}{4}{2}{0.3}{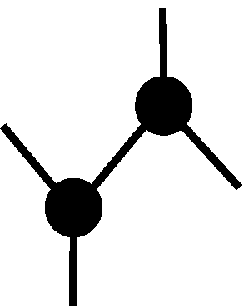}
\hspace{30pt}
\gph{-2}{4}{2}{0.3}{mu} = \gph{-2}{4}{2}{0.3}{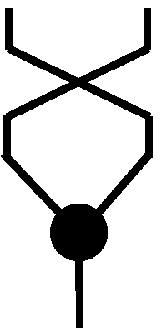}
\hspace{20pt}
\gph{-0.5}{2}{3}{0.3}{delta} = \gph{-3}{2}{3}{0.3}{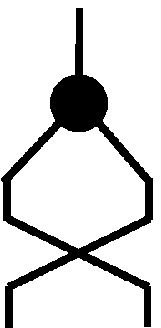}
\end{array}
\end{equation}
Because of these axioms, the representation of an $F$-graph is determined by its topological properties.

\begin{thm}[\cite{BCAK}]
Any two $F$-graphs $f$ and $g$ whose graphical representations are connected and which have the same numbers of inputs, outputs, and loops, are in fact the same. The number of loops represents the maximum number of wires that can be removed without destroying the connections of the representation.
\end{thm}

Here, for simplification, we provide notations for some $F$-graphs. A notation (called a {\em spider notation} in previous work \cite{BCAK}) is used for $F$-graphs that do not have any loops. An $F$-graph that has $m$ inputs and $n$ outputs is written as follows.
\begin{equation}
\gph{-12}{10}{11}{0.40}{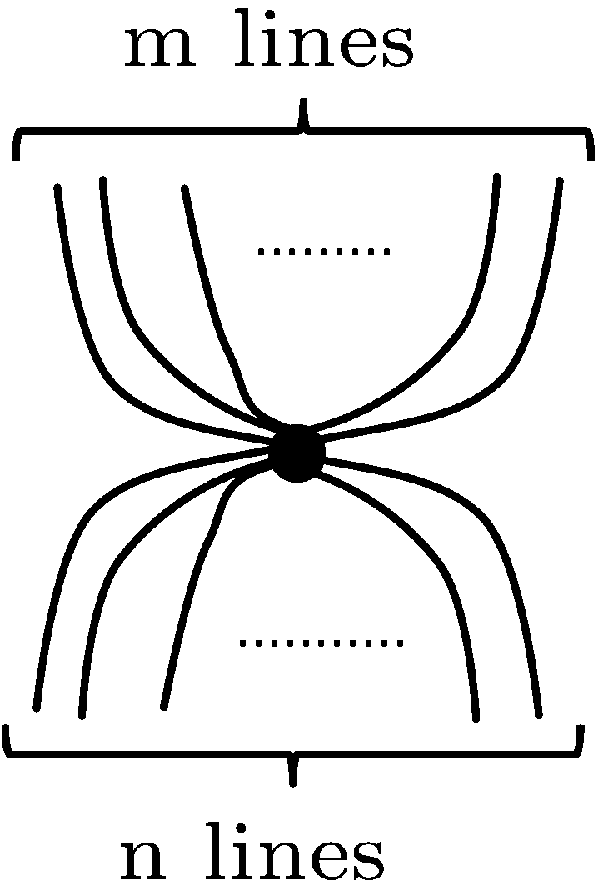}
\end{equation}
An $F$-graph that has no inputs or outputs is written as follows.
\begin{equation}
\gph{-5}{5}{5}{0.35}{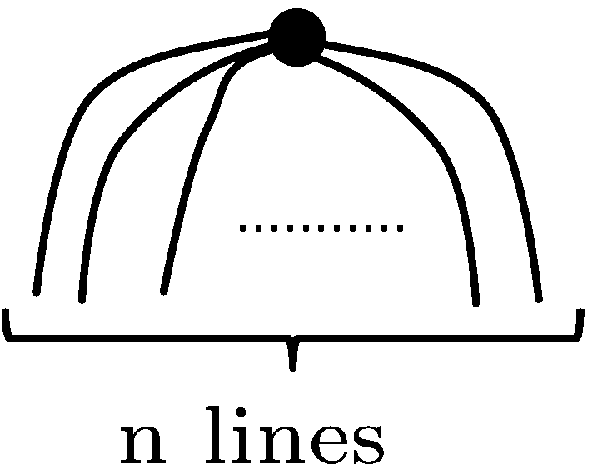}
\hspace{80pt}
\gph{-4}{5}{5}{0.35}{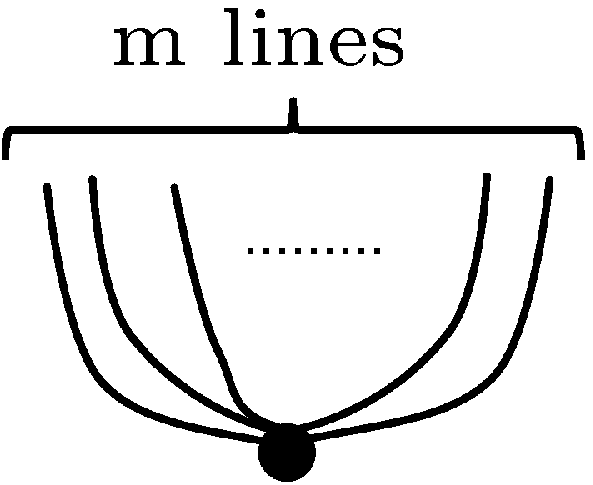}
\end{equation}
Moreover, some $F$-graphs that have exactly one loop have special notations. An $F$-graph that has one loop, no input, and one output is written as follows.
\begin{equation}
\gph{-2}{1.5}{1.5}{0.5}{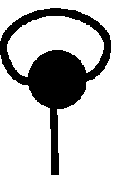}
\end{equation}
Similarly, an $F$-graph that has one loop, one input, and no outputs is expressed as follows.
\begin{equation}
\gph{-2}{1.5}{1.5}{0.5}{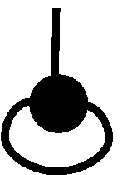}
\end{equation}
Finally, an $F$-graph that has one loop and no inputs or outputs is expressed as a circle.
\begin{equation}
\gph{-2}{1.5}{1.5}{0.5}{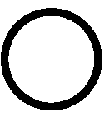}
\end{equation}

\subsection{Frobenius States}
CFAs correspond to Frobenius states, which require strong SLOCC-maximality and symmetry properties.

\begin{dfn}[Strong SLOCC-maximality]\rm
Let $\qbit{\Psi}$ be a tripartite state.
If there are $\qeff{\Phi_i}$ and $\qeff{\xi_i}$ such that the following holds, then $\qbit{\Psi}$ is strongly SLOCC-maximal:
\begin{equation}
\gphh{-8}{8}{8}{16}{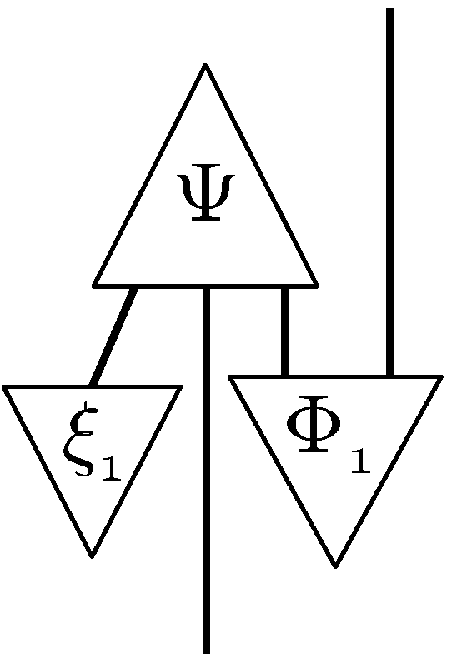}
=
\gphh{-8}{8}{8}{16}{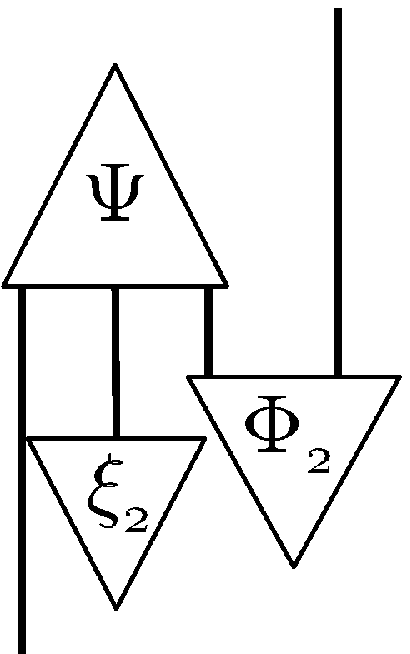}
=
\gphh{-8}{8}{8}{16}{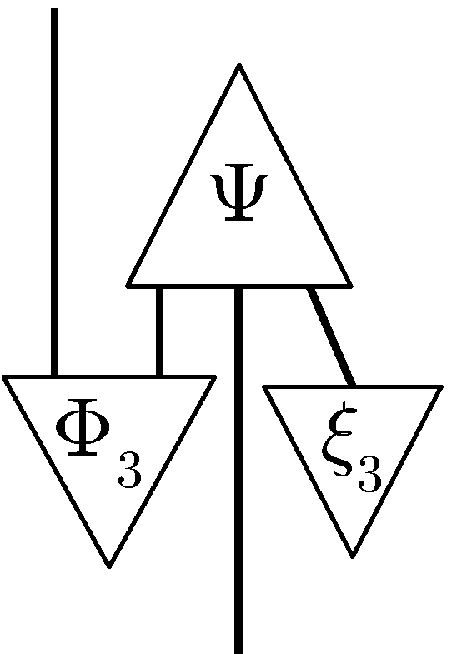}
=
\gphh{-8}{8}{8}{16}{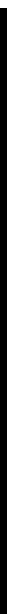}
\end{equation}
\end{dfn}
The following theorem shows the relation between SLOCC-maximality and strong SLOCC-ma\-xi\-ma\-li\-ty.
\begin{thm}[\cite{BCAK}]
Let $\qbit{\Psi}$ be a tripartite symmetric state. If $\qbit{\Psi}$ is strongly SLOCC-maximal, then $\qbit{\Psi}$ is SLOCC-maximal.
\end{thm}
\begin{dfn}[Symmetric State]\rm
Let $\qbit{\Psi}$ be a tripartite state.
If $\qbit{\Psi}$ satisfies the following equations, then $\qbit{\Psi}$ is a symmetric state: 
\begin{equation}
\gphh{-7}{5}{6}{13}{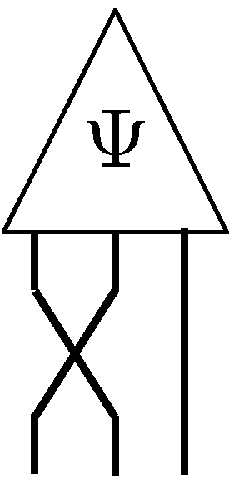}
=
\gphh{-7}{5}{6}{13}{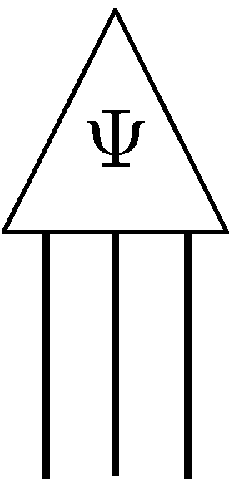}
=
\gphh{-7}{5}{6}{13}{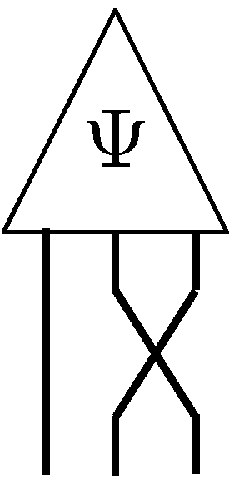}
\label{symmetry}
\end{equation}
\end{dfn}
$N$-partite symmetric states are defined similarly.
For symmetric states, the definition of strong SLOCC-maximality can be rewritten simply as the following proposition.

\begin{prp}[\cite{BCAK}]
Let $\qbit{\Psi}$ be a tripartite symmetric state.
$\qbit{\Psi}$ is strongly SLOCC-maximal
iff there are $\qeff{\Phi}$ and $\qeff{\xi}$ such that 
\begin{equation}
\gphh{-8}{8}{8}{16}{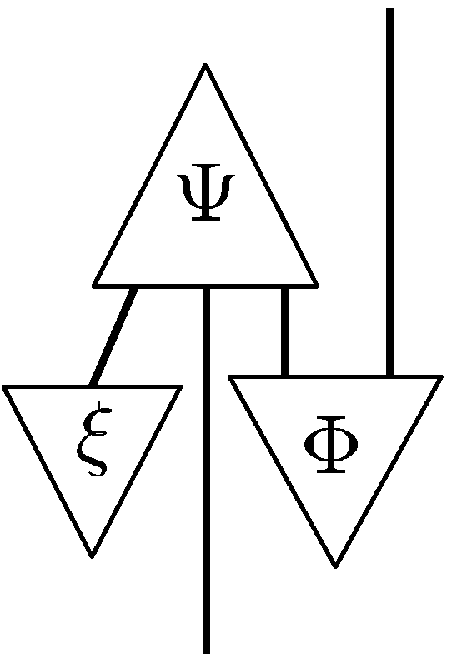}
=
\gphh{-8}{8}{8}{16}{linemax}
\label{sloccmax}
\end{equation}
\end{prp}

For any strongly SLOCC-maximal and tripartite symmetric state $\qbit{\Psi}$, if $\qeff{\xi}$ is given, then $\qeff{\Phi}$ is uniquely determined. The converse is also true.

\begin{prp}[\cite{BCAK}]\label{uniqueness}
Let $\qbit{\Psi}$ be a strongly SLOCC-maximal and tripartite symmetric state. Each of a pair of states $\qeff{\Phi}$ and $\qeff{\xi}$ satisfying (\ref{sloccmax}) is uniquely determined by the other.
\end{prp}
\begin{proof}
Suppose ($\qbit{\Psi}$, $\qeff{\Phi}$, $\qeff{\xi}$) and ($\qbit{\Psi}$, $\qeff{\Phi}$, $\qeff{\xi'}$) satisfy (\ref{sloccmax}).
\begin{equation}
\gph{-4}{4}{4}{0.4}{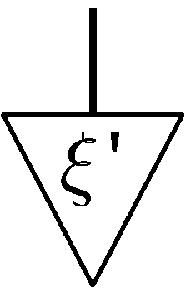}
= \gph{-8}{8}{8}{0.4}{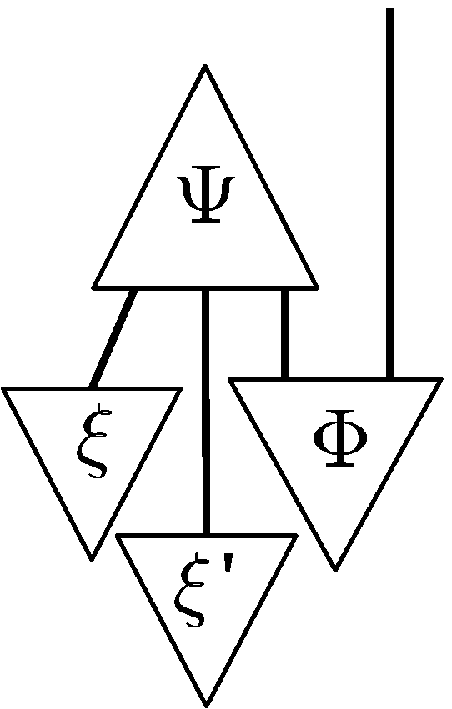}
= \gph{-8}{8}{8}{0.4}{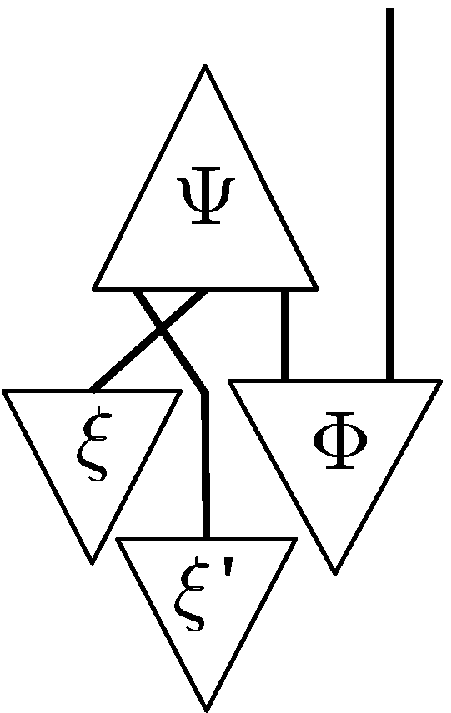}
= \gph{-8}{8}{8}{0.4}{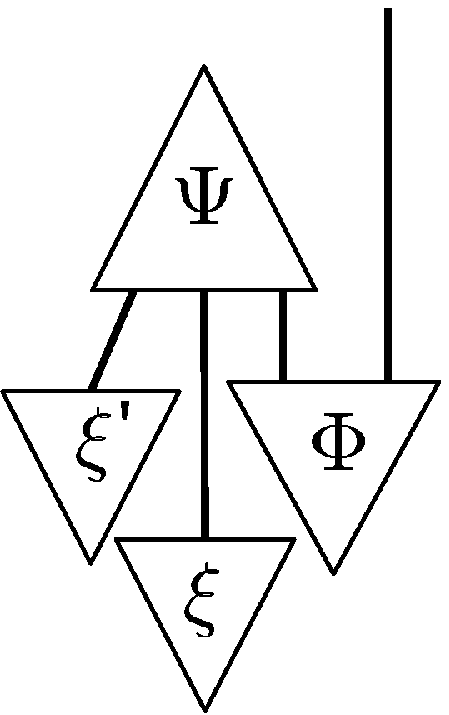}
= \gph{-4}{4}{4}{0.4}{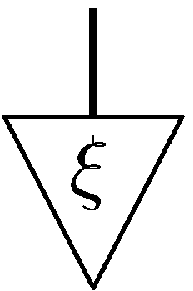}
\end{equation}
Similarly, assume ($\qbit{\Psi}$, $\qeff{\Phi}$, $\qeff{\xi}$) and ($\qbit{\Psi}$, $\qeff{\Phi'}$, $\qeff{\xi}$) satisfy (\ref{sloccmax}).
\begin{equation}
\gph{-4}{4}{4}{0.4}{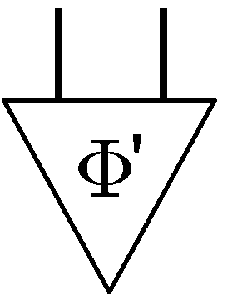}
= \gph{-10}{10}{9}{0.4}{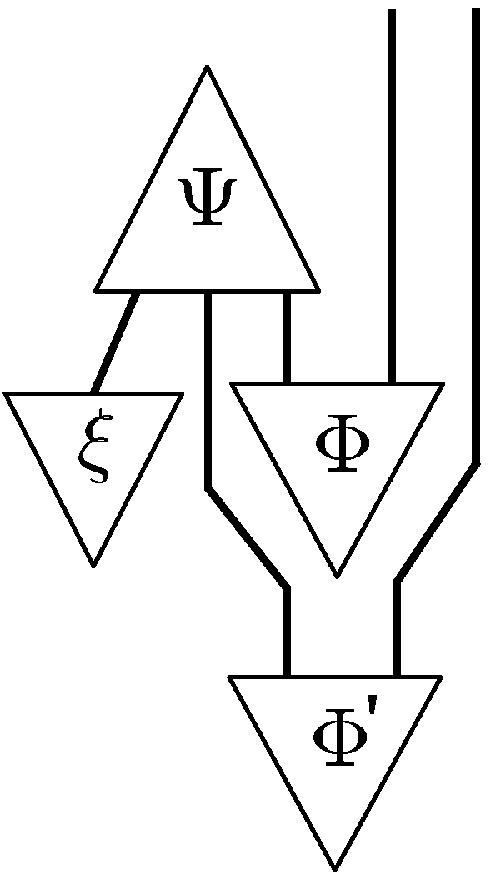}
= \gph{-9}{9}{9}{0.4}{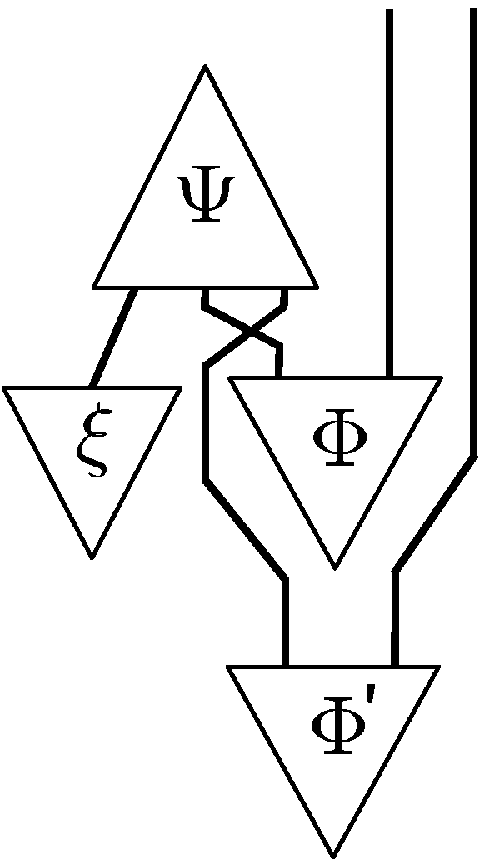}
= \gph{-9}{9}{9}{0.4}{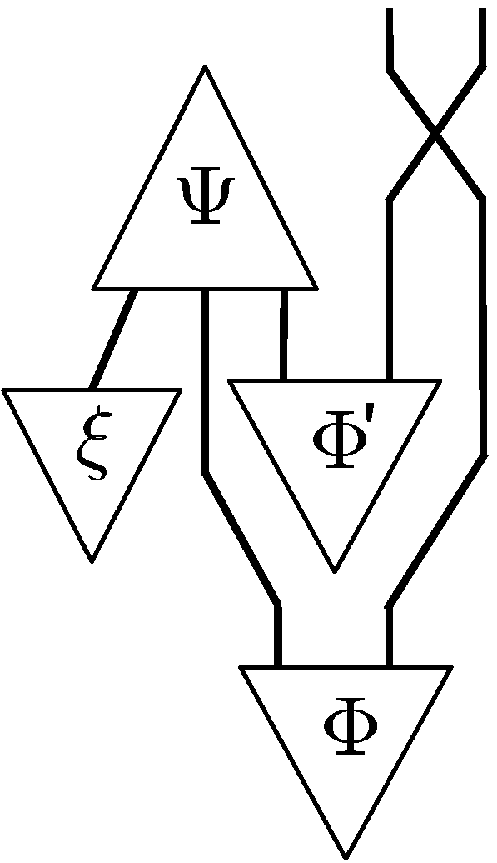}
= \gph{-4}{4}{4}{0.4}{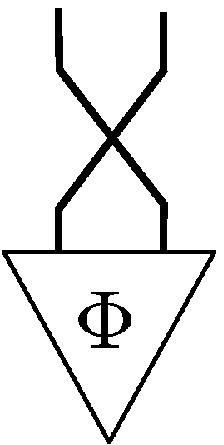}
\end{equation}
Then, substituting $\qbit{\Phi}$ for $\qbit{\Phi'}$,
\begin{equation}
\gph{-4}{4}{4}{0.4}{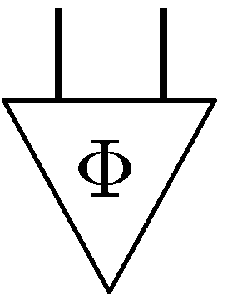}
= \gph{-4}{4}{4}{0.4}{alphaphi}
= \gph{-4}{4}{4}{0.4}{pphi}.
\end{equation}
\end{proof}

Frobenius states also require not only symmetry but also strong symmetry to correspond to CFAs.

\begin{dfn}[Strong Symmetry]\rm
Let $\qbit{\Psi}$ be a tripartite state.
If $\qbit{\Psi}$ is a symmetric state and there is $\qeff{\Phi}$ such that 
\begin{equation}
\gphh{-10}{10}{10}{20}{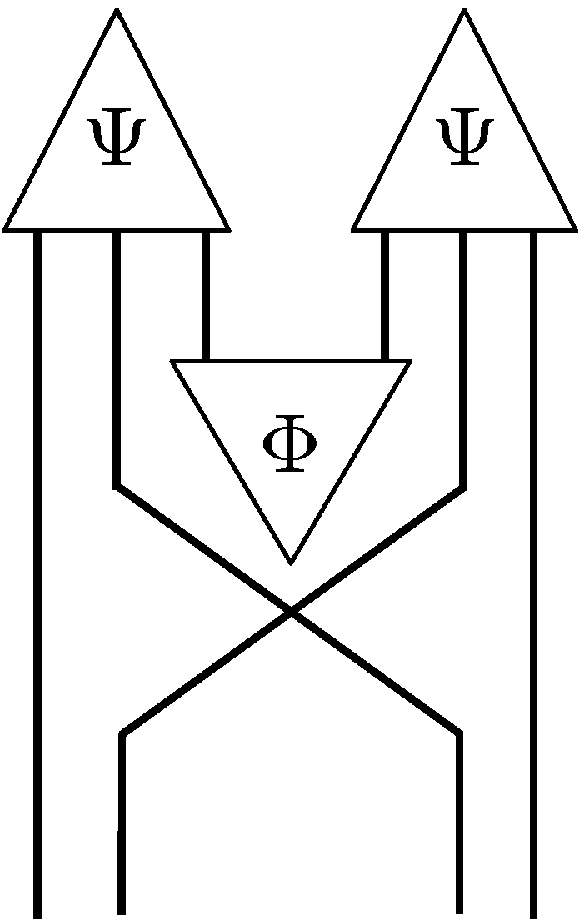}
=
\gphh{-10}{10}{10}{20}{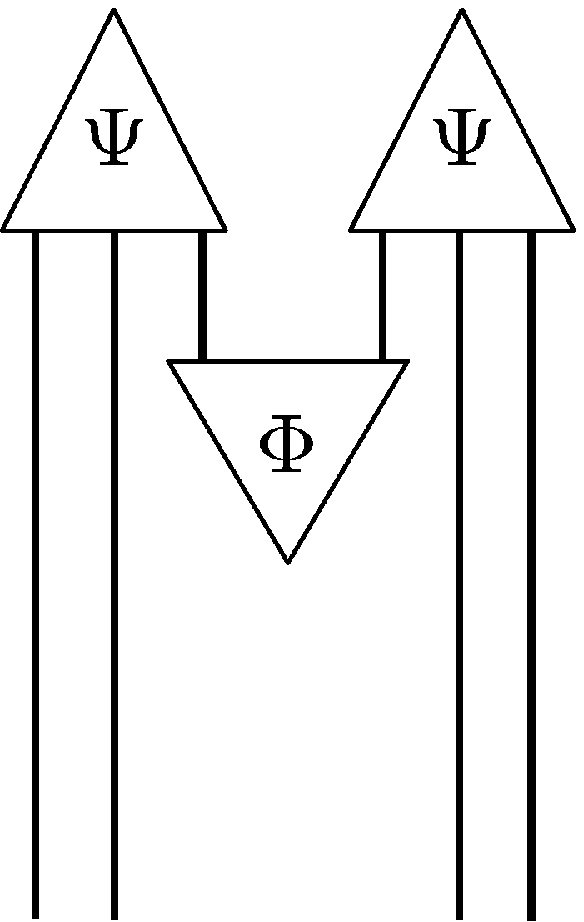}
\label{strongsym}
\end{equation}
then $\qbit{\Psi}$ is strongly symmetric.
\end{dfn}

\begin{dfn}[Frobenius State]\rm
Let $\qbit{\Psi}$ be a tripartite state.
If there are $\qeff{\Phi}$ and $\qeff{\xi}$ such that they satisfy (\ref{sloccmax}) and (\ref{strongsym}), then $\qbit{\Psi}$ is a Frobenius state.
\end{dfn}

Notice that this definition requires all equations with the same $\qeff{\Phi}$ and $\qeff{\xi}$ to be satisfied.
According to Proposition \ref{uniqueness}, we write a Frobenius state $\qbit{\Psi}$ with $\qeff{\xi}$ to indicate the triple $\qbit{\Psi}$, $\qeff{\Phi}$, and $\qeff{\xi}$ such that they satisfy the Frobenius conditions.
Note that for a Frobenius state $\qbit{\Psi}$, more than one pair ($\qeff{\xi}$, $\qeff{\Phi}$)  is generally possible such that $\qbit{\Psi}$, $\qeff{\Phi}$, and $\qeff{\xi}$ satisfy the Frobenius conditions.

Frobenius states and CFAs are strictly connected. The connection is given by the following theorems.

\begin{thm}[\cite{BCAK}]\label{CFAtoState}
For any CFA,
\begin{equation}
\begin{array}{ccc}
\qbit{\Psi}:= \gphh{-1.5}{1.5}{1.5}{4}{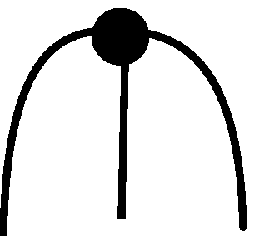}

&
\qeff{\Phi}:= 
\gphh{-1.5}{1.5}{1.5}{3.5}{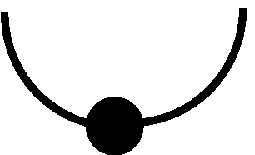}
&
\qeff{\xi}:= \gepsilon
\end{array}
\end{equation}
$\qbit{\Psi}$ is a Frobenius state with $\qeff{\Phi}$ and $\qbit{\xi}$.
\end{thm}

\begin{thm}[\cite{BCAK}]\label{StatetoCFA}
Any Frobenius state defines a CFA $(\Hilbert, \gmu, \gdelta, \geta, \gepsilon)$ as
\begin{equation}
\begin{array}{llll}
\gmu:= \gph{-5}{4}{5}{0.3}{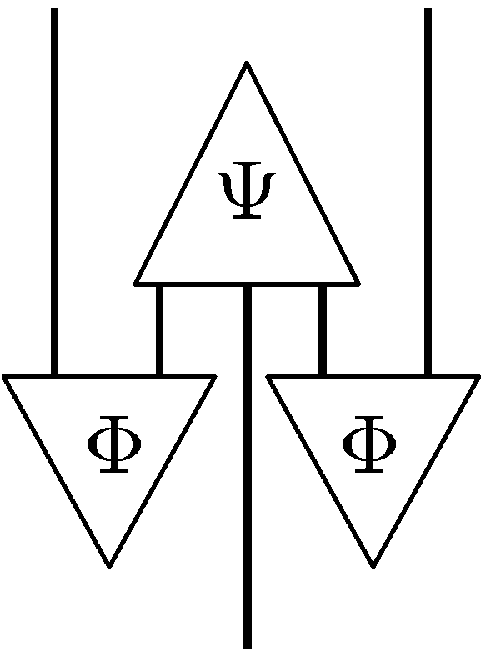}
&
\geta:= \gph{-5}{4}{5}{0.3}{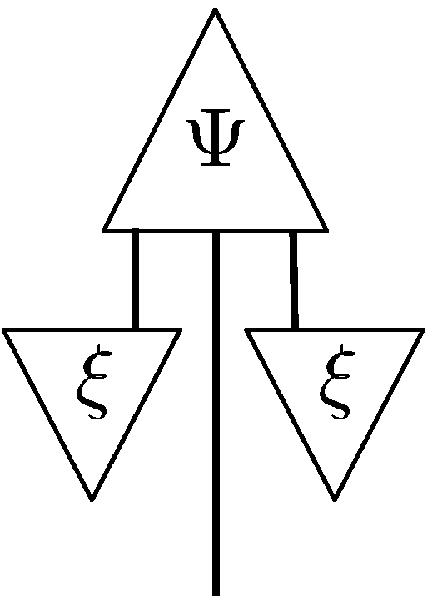}
&
\gdelta:= \gph{-5}{8}{5}{0.3}{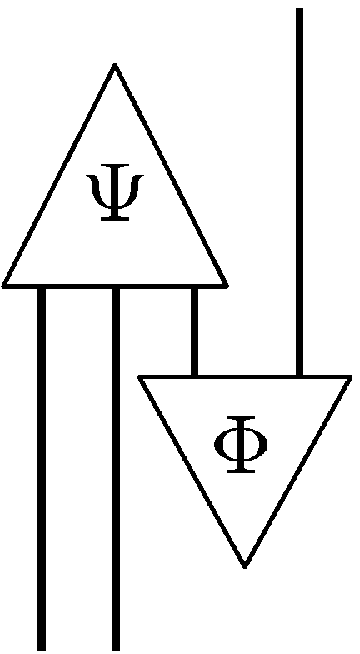}
&
\gepsilon:= \gph{-3}{3}{3}{0.4}{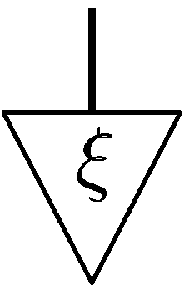}
\end{array}
\end{equation}
\end{thm}

Theorem \ref{CFAtoState} describes how to create a Frobenius state from a CFA, and Theorem \ref{StatetoCFA} describes the converse.
Note that a CFA induced by a Frobenius state with $\qeff{\xi}$ and $\qeff{\Phi}$ is the same as the original CFA. For example, we can show CFAs that correspond to tripartite qubits; these CFAs are defined in $\mathbf{FdHilb}$.
In tripartite qubits, six SLOCC classes are possible: $\overline{\qbit{000}}$, $\overline{\qbit{000}+\qbit{011}}$, $\overline{\qbit{000}+\qbit{101}}$, $\overline{\qbit{000}+\qbit{110}}$, $\overline{\qbit{\mbox{GHZ}}}$, and $\overline{\qbit{\mbox{W}}}$.
Obviously the first four classes are not strongly SLOCC-maximal. In contrast, $\qbit{\mbox{GHZ}}$ and $\qbit{\mbox{W}}$, which were defined in (\ref{GHZ}) and (\ref{W}), are Frobenius states, so these states correspond to CFAs.

\begin{exm}
$\qbit{GHZ}$ with $\qeff{\xi}:= \qeff{0}+\qeff{1}$ induces a CFA 
\begin{equation}
\begin{array}{llll}
\gmu:= \qact{0}{00} + \qact{1}{11} 
& \geta:= \qbit{0}+\qbit{1}
& \gdelta:= \qact{00}{0} + \qact{11}{1}
& \gepsilon:= \qeff{0}+\qeff{1}
\end{array}
\label{GHZCFA}
\end{equation}
\end{exm}

\begin{exm}
$\qbit{W}$ with $\qeff{\xi}:= \qeff{0}$ induces a CFA
\begin{equation}
\begin{array}{llll}
\gmu:= \qact{0}{01} + \qact{0}{10} + \qact{1}{11} & \geta:= \qbit{1} &
\gdelta:= \qact{00}{0} + \qact{01}{1} +\qact{10}{1} & \gepsilon:= \qeff{0}
\end{array}
\label{WCFA}
\end{equation}
\end{exm}

\subsection{Classification of Tripartite Qubits}

In \cite{BCAK}, two kinds of CFAs were defined to classify tripartite
qubit states.

\begin{dfn}[Special Commutative Frobenius Algebra]\rm
A commutative Frobenius algebra that satisfies the following equation is called a {\em special commutative Frobenius algebra (SCFA)}:
\begin{equation}
\gphh{-3.5}{3.1}{2.4}{8}{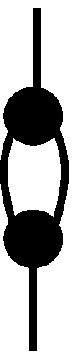}
=
\gphh{-3.5}{3.1}{2.4}{8}{line}.
\label{SCFAcond}
\end{equation}
\end{dfn}

\begin{dfn}[Anti-special Commutative Frobenius Algebra]\rm
A commutative Frobenius algebra that satisfies the following equation is called an {\em anti-special commutative Frobenius algebra (ACFA)}:
\begin{equation}
\gphh{-3.5}{3}{2.5}{8}{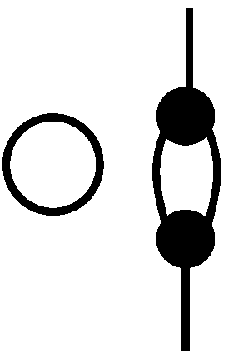}
=
\gphh{-3.5}{3}{2.5}{8}{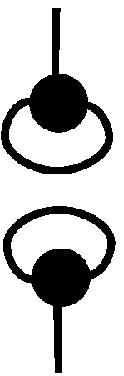}.
\label{ACFAcond}
\end{equation}
\end{dfn}

These algebras are topologically different from each other. Using simple calculations, it is obvious that  (\ref{GHZCFA}) is an SCFA and that (\ref{WCFA}) is an ACFA.
For distinction, an SCFA is expressed as a white dot $\gph{-0.6}{2}{0}{0.24}{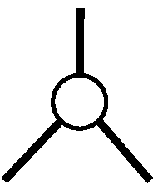}$ and an ACFA is expressed as a black dot $\gph{-0.6}{2}{0}{0.24}{delta}$. 
In \cite{BCAK}, it was demonstrated that these two types of CFAs strictly correspond to the two SLOCC classes of tripartite qubits.

\begin{thm}[\cite{BCAK}]
Let $\qbit{\Psi}$ be a Frobenius state.
$\qbit{\Psi}$ is SLOCC-equivalent to the GHZ state iff there is $\qeff{\xi}$ such that $\qbit{\Psi}$ with $\qeff{\xi}$ induces an SCFA.
$\qbit{\Psi}$ is SLOCC-equivalent to the $W$ state if and only if there is $\qeff{\xi}$ such that $\qbit{\Psi}$ with $\qeff{\xi}$ induces an ACFA.
\end{thm}

%-------------------
% Section 4.
%-------------------
\section{Qutrits and Commutative Frobenius Algebras}
In this section, we classify three-dimensional CFAs, show correspondence between tripartite qutrits and CFAs, and demonstrate how to compose any qutrit graphically using arguments similar to those used in \cite{BCAK}.
However, unlike in the case of qubits, infinitely many SLOCC classes are possible in tripartite qutrits. 
Hence, we must distinguish between SLOCC classes that include Frobenius states and those that do not. 
To this end, we use the requirements of Frobenius states, i.e., strong SLOCC-maximality, symmetry, and strong symmetry. First, we identify any class that does not have a SLOCC-maximal state. Then we examine which classes include a symmetric state. Next, we classify Frobenius states using the strong symmetry condition, and then define three CFAs that correspond to Frobenius states, classify these using graphical equations, and prove that the classification strictly corresponds to the three CFAs. Finally, we represent any qutrits graphically.

\subsection{Non-Maximal Class}
First, we use the first condition, strong SLOCC-maximality. It requires a tripartite qutrit to have a full rank density matrix in each single qutrit. Some SLOCC classes do not include strong SLOCC-maximal states. The absence of states in a SLOCC class determines the properties of the class. We call a SLOCC class that does not include any strong SLOCC-maximal state a \textit{non-maximal class}.
\begin{lmm}
For any tripartite SLOCC class X, if X includes a tripartite qutrit that is not strongly SLOCC-maximal, then X is a non-maximal class.
\end{lmm}
\begin{proof}
Let $\qbit{\phi}$ be a tripartite qutrit that is SLOCC-equivalent to a strongly SLOCC-maximal state $\qbit{\psi}$. $\qbit{\psi}$ has $\qeff{\xi_i}$ and $\qeff{\Phi_i}$ $(i \in \{1, 2, 3\})$ such that they satisfy the SLOCC-maximal conditions. Because $\qbit{\phi}$ and $\qbit{\psi}$ are SLOCC-equivalent, there are invertible matrices $L_1$, $L_2$, and $L_3$ such that $\qbit{\phi} = (L_1 \otimes L_2 \otimes L_3)\qbit{\psi}$. $\qeff{\xi_i'}:= \qeff{\xi_i} \circ L_i^{-1}$ and $\qeff{\Phi_i'}:= \qeff{\Phi_i} \circ(L_j^{-1} \otimes L_k^{-1})$ with $i, j, k\in \{1,2,3\}$, which differ from each other. $\qbit{\phi}$, $\qeff{\xi_i'}$, and $\qeff{\Phi_i'}$ satisfy the SLOCC-maximal conditions.
\end{proof}

Frobenius states require strong SLOCC-maximality, so a non-maximal class does not have Frobenius states. Using simple calculations, we can prove that for any $i \in \{0, \ldots, 24\}$, $\overline{\qbit{\psi_i}}$ is a non-maximal class.

In $\overline{\qbit{\pi(\phi, \varphi, \chi, \psi)}}$, if $\qbit{\phi}$ and $\qbit{\chi}$ can be expressed as $\qbit{\phi} = \alpha\qbit{0}+\beta\qbit{1}$ and $\qbit{\chi} = \gamma\qbit{0}+\delta\qbit{1}$ by some complex numbers $\alpha$, $\beta$,$\gamma$, and $\delta$, then this class does not include a strong SLOCC-maximal state. The same can be said for $\qbit{\varphi}$ and $\qbit{\psi}$.

\subsection{Non-Symmetric Class}
Next, we use the second condition, i.e., symmetry. Many SLOCC classes include strong SLOCC-maximal states, but a few of them include symmetric states. A SLOCC class with strong SLOCC-maximal states but no symmetric states is called a \textit{non-symmetric class}. The following lemma is used to identify non-symmetric classes. 
\begin{lmm}\label{permlmm}
For any permutation $P$, if a tripartite qutrit $\qbit{\phi}$ is SLOCC-equivalent to a tripartite symmetric state $\qbit{\psi}$, $P\qbit{\phi}$ is SLOCC-equivalent to $\qbit{\phi}$. \qed
\end{lmm}
Applying this lemma to representations, we can prove that for any $i \in \{0, \ldots , 9\}$, $\overline{\qbit{\phi_i}}$ is a non-symmetric class. 
In addition to Lemma \ref{permlmm}, any permutation of two qutrits can be represented in a $3 \times 3$ matrix.
\begin{lmm}\label{ilolmm}
For any $N$-partite qutrit $\qbit{\phi}$ and any permutation $P$ between an $i$th qutrit and a $j$th qutrit, if $\qbit{\phi}$ is SLOCC-equivalent to an $N$-partite symmetric state $\qbit{\psi}$, then there is a $3 \times 3$ invertible matrix $L$ and
\begin{equation}
M_k = \left\{\begin{array}{cc}
L & (k = i)\\
L^{-1} & (k = j)\\
I & (otherwise)
\end{array}\right.
\end{equation}
such that $P\qbit{\phi} = (\otimes_{k=1}^N M_k) \qbit{\phi}$. Here $I$ is the identity matrix.
\end{lmm}
\begin{proof}
Let $\qbit{\phi}$ be a tripartite qutrit that is SLOCC-equivalent to a tripartite symmetric state $\qbit{\psi}$.
There are invertible matrices $L_k$ such that $\qbit{\psi} = \otimes_{k = 1}^N L_k \qbit{\phi}$. Let $F_k$ be a function such that
\begin{equation}
F_k:= \left\{\begin{array}{cc}
L_j^{-1} & (k = i)\\
L_i^{-1} & (k = j)\\
L_k^{-1} & (otherwise)
\end{array}\right.
\end{equation}
This satisfies $P\qbit{\phi} = \otimes_{k=1}^N F_k \qbit{\psi}$. Let $L = L_j^{-1} \circ L_i$, then the $M_k$ defined in this lemma satisfies $P\qbit{\phi} = \otimes_{k=1}^N (F_k \circ L_k) \qbit{\phi} = \otimes_{k=1}^N M_k \qbit{\phi}$.
\end{proof}

Using these lemmas, we can consider all classes expressed as $\overline{\qbit{\pi(\phi, \varphi, \chi, \psi)}}$. We pick up a SLOCC-maximal class $\overline{\qbit{\pi(\phi, \varphi, \chi, \psi)}}$. 
By the above arguments and calculation, it is divided into two cases, i.e., $\overline{\qbit{\pi(\phi', \varphi', \qbit{2}, \qbit{2})}}$ and $\overline{\qbit{\pi(\phi', \qbit{2}, \qbit{2}, \alpha\qbit{0}+\beta\qbit{1})}}$.

First, we assume that $\qbit{\pi(\phi', \qbit{2}, \qbit{2}, \psi')}$ belongs to the class such that $\qbit{\psi'} = \alpha\qbit{0}+\beta\qbit{1}$. Consider a permutation between the second and third qutrits. According to Lemma \ref{ilolmm}, there is an invertible matrix $L$ such that $\qbit{\pi(\phi', \qbit{2}, \qbit{2}, \psi')} = (I \otimes L \otimes L^{-1})\qbit{\pi(\phi', \qbit{2}, \qbit{2}, \psi')}$. 
Performing this calculation, we get $\alpha = \beta = 0$.
As a result, $\overline{\qbit{\pi(\phi', \qbit{2}, \qbit{2}, \psi')}}$ is a non-symmetric class.

Second, we consider $\qbit{\pi(\phi', \varphi', \qbit{2}, \qbit{2})}$. $\qbit{\phi'}$ and $\qbit{\varphi'}$ are 
\begin{eqnarray}
\qbit{\phi'} &=& \alpha\qbit{0}+\beta\qbit{1}+\gamma\qbit{2}\\
\qbit{\varphi'} &=& \delta\qbit{0}+\eta\qbit{1}+\theta\qbit{2}
\end{eqnarray}
Using a similar calculation, we arrive at $\gamma = \theta = 0$. 

Using the same process, we can check whether the classes 
$\overline{\qbit{\varphi_1}}$, $\overline{\qbit{\varphi_2}}$, and $\overline{\qbit{\varphi_3}}$
are non-symmetric.

We have identified all non-symmetric classes. The rest of the classes are 
$\overline{\qbit{\G}}$,
$\overline{\qbit{w_0}}$, 
$\overline{\qbit{s_0}}$, 
$\overline{\qbit{s_1}}$,
$\overline{\qbit{\pi(\phi', \varphi', \qbit{2}, \qbit{2})}}$. 
There are symmetric states in the following classes: 
$\qbit{\G}$, 
$\qbit{\W}:= \qbit{002}+\qbit{011}+\qbit{020}+\qbit{101}+\qbit{110}+\qbit{200}$,
$\qbit{s_2}:= \qbit{000}+\qbit{012}+\qbit{021}+\qbit{102}+\qbit{120}+\qbit{201}+\qbit{210}$,
$\qbit{s_3}:= \qbit{012}+\qbit{021}+\qbit{102}+\qbit{120}+\qbit{201}+\qbit{210}$,
$\qbit{\I}:= \qbit{001}+\qbit{010}+\qbit{100}+\qbit{222}$.

For $\qbit{\pi(\phi', \varphi', \qbit{2}, \qbit{2})}$, there are two cases that are SLOCC-equivalent to $\qbit{000}+\qbit{011}+\qbit{100}+\qbit{222}$ and $\qbit{000}+\qbit{011}+\qbit{101}+\qbit{222}$. The first case is the same class as $\qbit{\G}$. In the second class, there is a symmetric state $\qbit{\I}$.

\subsection{Frobenius Class}

A SLOCC class that includes a Frobenius state is called a \textit{Frobenius class}.
We already know that only five classes include symmetric states. To restrict the classes to Frobenius classes, we can use the following theorem.
\begin{thm}
For any tripartite symmetric states $\qbit{\phi}$ and $\qbit{\psi}$, if they are SLOCC-equivalent, then there is a $3 \times 3$ matrix $L$ such that $\qbit{\phi} = (L \otimes L \otimes L)\qbit{\psi}$.
\end{thm}
\begin{proof}
This can be proved using arguments similar to those used in \cite{PhysRevA.81.052315}.
Suppose $\qbit{\phi}$ and $\qbit{\psi}$ are connected by a non-symmetric transformation such as $\qbit{\phi} = (L_1 \otimes L_2 \otimes L_3)\qbit{\psi}$. There is $B$ such that $\qbit{\psi} = (B \otimes B^{-1} \otimes I)\qbit{\psi}$.
$B$ may be diagonalizable, converted into a Jordan block, or converted into two Jordan blocks. Moreover, the first two cases are divided by their eigenvalues. 
In any case, we can prove that there is $L$ such that $\qbit{\psi} = (L \otimes L \otimes L)\qbit{\psi_0}$, where $\qbit{\psi_0}$ is one of $\qbit{000}$, $\qbit{000}+\qbit{111}$, $\qbit{001}+\qbit{010}+\qbit{111}$, $\qbit{\G}$, $\qbit{\W}$, and $\qbit{\I}$. We know that these classes are not SLOCC-equivalent to each other. Then, $\qbit{\psi}$ can be converted into $\qbit{\phi}$ by symmetric transformation via $\qbit{\psi_0}$.
\end{proof}
As a reminder, any symmetric state that is SLOCC-equivalent to a Frobenius state is also a Frobenius state. Hence, we only need to check for a tripartite symmetric state in all symmetric classes to judge whether or not the classes are Frobenius classes. Using simple calculations, we can obtain three Frobenius states:

$\qbit{\G}$ with $\qeff{\xi}:= \qeff{0} + \qeff{1} + \qeff{2}$:
\begin{eqnarray}
\begin{array}{ll}
\gmu:= \qact{0}{00} + \qact{1}{11}+ \qact{2}{22}
&
\geta:= \qbit{0} + \qbit{1} + \qbit{2}
\\
\gdelta:= \qact{00}{0} + \qact{11}{1}+ \qact{22}{2}
&
\gepsilon:= \qeff{0} + \qeff{1} + \qeff{2}
\end{array}
\end{eqnarray}

$\qbit{\W}$ with $\qeff{\xi}:= \qeff{0}$:
\begin{eqnarray}
\begin{array}{ll}
\gmu:= \qact{0}{02}+\qact{0}{11}+\qact{0}{20}+\qact{1}{12}+\qact{1}{21}+ \qact{2}{22}
& \geta:= \qbit{2}
\\
\gdelta:= \qact{00}{0} + \qact{01}{1}+\qact{10}{1}+\qact{02}{2}+\qact{11}{2}+\qact{20}{2}
& \gepsilon:= \qeff{0}
\end{array}
\end{eqnarray}

$\qbit{\I}$ with $\qeff{\xi}:= \qeff{0} + \qeff{2}$:
\begin{eqnarray}
\begin{array}{ll}
\gmu:= \qact{0}{01} + \qact{0}{10} + \qact{1}{11}+ \qact{2}{22}
&
\geta:= \qbit{1} + \qbit{2}
\\
\gdelta:= \qact{00}{0} + \qact{01}{1} + \qact{10}{1} + \qact{22}{2}
&
\gepsilon:= \qeff{0} + \qeff{2}
\end{array}
\end{eqnarray}
We call these algebras $\G$, $\W$, and $\I$.

However, the other two classes are not Frobenius classes. 
First, we consider $\qbit{s_2}$. Let $\qeff{\xi} = \alpha\qeff{0} + \beta\qeff{1} + \gamma\qeff{2}$, then 
$\qeff{\Phi} = 
\frac{1}{\alpha(\alpha^2-2\beta\gamma)}(\alpha^2\qeff{00}
-\alpha\beta\qeff{10}
-\alpha\gamma\qeff{20}
-\alpha\beta\qeff{01}
+\beta^2\qeff{11}
+(\alpha^2-\beta\gamma)\qeff{21}
-\alpha\gamma\qeff{02}
+(\alpha^2-\beta\gamma)\qeff{12}
+\gamma^2\qeff{22})
$. Calculating (\ref{strongsym}), we know that $\qbit{s_2}$ does not have strong symmetry.

Second, we consider $\qbit{s_3}$. Similarly, we let $\qeff{\xi} = \alpha\qeff{0}+\beta\qeff{1}+\gamma\qeff{2}$. A $\qeff{\Phi}$ that satisfies (\ref{sloccmax}) is 
$
-\frac{\alpha}{2\beta\gamma}\qeff{00}
+\frac{1}{2\gamma}\qeff{01}
+\frac{1}{2\beta}\qeff{02}
+\frac{1}{2\gamma}\qeff{10}
-\frac{\beta}{2\alpha\gamma}\qeff{11}
+\frac{1}{2\alpha}\qeff{12}
+\frac{1}{2\beta}\qeff{20}
+\frac{1}{2\alpha}\qeff{21}
-\frac{\gamma}{2\alpha\beta}\qeff{22}
$. The $\qeff{\xi}$, $\qeff{\Phi}$, and $\qbit{s_3}$ do not have strong symmetry.

As a result, we obtained three Frobenius classes and proved that the other classes are not Frobenius classes.

\subsection{Classification of Commutative Frobenius Algebras}

We can judge which classes are SCFA or ACFA by calculating $\gph{-1.4}{3}{2}{0.2}{dual}$. We can use the language of smooth manifolds to verify that $\G$ is an SCFA and $\W$ is an ACFA, but $\I$ is neither an SCFA nor an ACFA. Next, we define intermediate special commutative Frobenius algebras.
\begin{dfn}[ISCFA]\rm
A commutative Frobenius algebra that satisfies the following two equations is an {\em intermediate special commutative Frobenius algebra (ISCFA)}:
\begin{eqnarray}
\gphh{-7}{7}{9}{15}{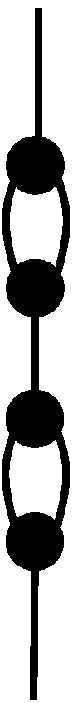} &=& \gphh{-7}{7}{9}{15}{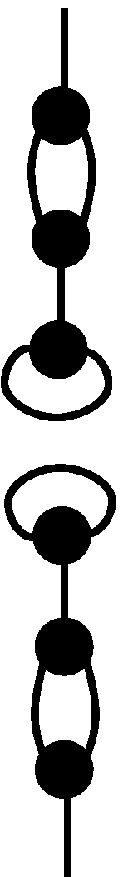} \label{ISCFAcond1}
\\
\gphh{-3.4}{4}{3}{9}{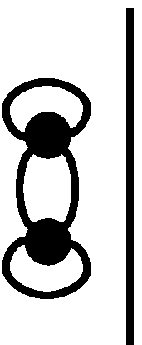} &=& \gphh{-3.4}{4}{3}{9}{line} \label{ISCFAcond2}
\end{eqnarray}.
\end{dfn}
An ISCFA is expressed by a white dot with a central small black dot $\gph{-0.2}{2.5}{0}{0.23}{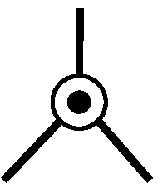}$.
We can immediately verify that $\I$ is an ISCFA, and that $\G$ and $\W$ are not. Moreover, we can prove that these three algebras correspond exactly to the three Frobenius classes.

\begin{thm}\label{slocccond}
If a Frobenius state $\qbit{\psi}$ is SLOCC-equivalent to a Frobenius state $\qbit{\phi}$ that induces an SCFA, an ACFA, and an ISCFA with $\qeff{\xi}$, then there is $\qeff{\xi'}$ such that $\qbit{\psi}$ with $\qeff{\xi'}$ induces an SCFA, an ACFA, and an ISCFA, respectively.
\end{thm}
\begin{proof}
Let $\mathcal{X} = (\C^3, \gph{-1.5}{2}{1}{0.25}{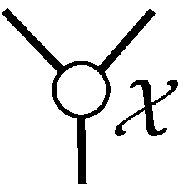}, \gph{-0.7}{1}{1}{0.25}{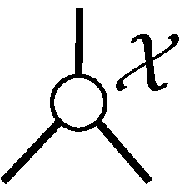}, \gph{-1.0}{2}{1}{0.25}{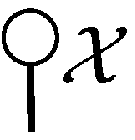}, \gph{-1.2}{2}{1}{0.25}{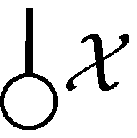})$ be an SCFA induced by $\qbit{\phi}$ with $\qeff{\xi}$.
There is an invertible matrix $L$ such that $\qbit{\psi} = (L \otimes L \otimes L)\qbit{\phi}$. $\qbit{\psi}$ with $\qeff{\xi'}:= \qeff{\xi} \circ L^{-1}$ induces a CFA $\mathcal{A} = (\C^3, \gph{-1.5}{2}{1}{0.25}{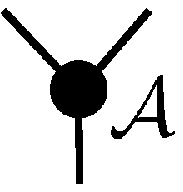}, \gph{-0.7}{1}{1}{0.25}{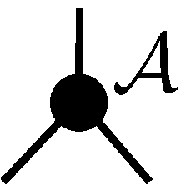}, \gph{-1.5}{2}{1}{0.25}{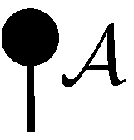}, \gph{-1.5}{2}{1}{0.25}{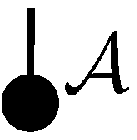})$:
\begin{equation}
\begin{array}{lclclcl}
\gph{-1.5}{2}{1}{0.3}{muA}:= \gph{-4.5}{6}{5}{0.3}{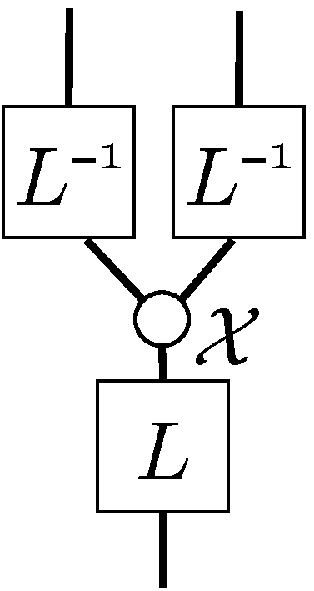}
&&
\gph{-1.5}{2}{1}{0.3}{etaA}:= \gph{-2}{5}{2}{0.3}{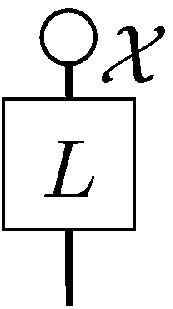}
&&
\gph{-0.5}{1}{1}{0.3}{deltaA}:= \gph{-4.5}{6}{5}{0.3}{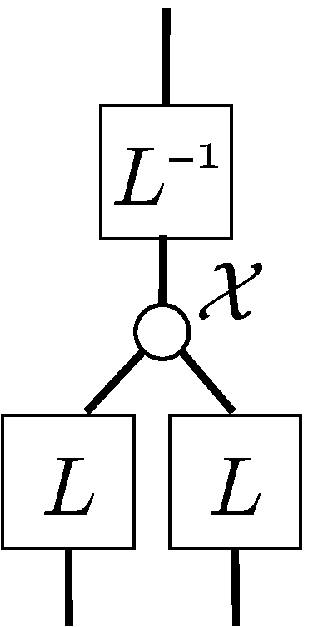}
&&
\gph{-1.0}{2}{1}{0.3}{epsilonA}:= \gph{-3}{5}{2}{0.3}{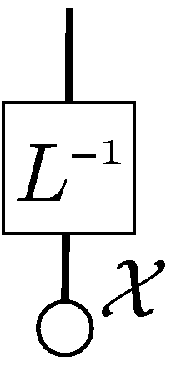}
\end{array}
\end{equation}
$\mathcal{A}$ satisfies the SCFA condition.
The remaining cases are proved in the same way.
\end{proof}

\begin{cor}
If a Frobenius state $\qbit{\psi}$ is SLOCC-equivalent to $\qbit{\G}$, $\qbit{\W}$, and $\qbit{\I}$, then there is $\qeff{\xi'}$ such that $\qbit{\psi}$ with $\qeff{\xi'}$ induces an SCFA, an ACFA, and an ISCFA, respectively.
\end{cor}

\begin{thm}\label{SCFAuniq}
Let $\qbit{\psi}$ be a Frobenius state.
If there is $\qeff{\xi}$ such that $\qbit{\psi}$ with $\qeff{\xi}$ induces an SCFA, then $\qbit{\psi}$ is SLOCC-equivalent to $\qbit{\G}$.
\end{thm}
\begin{proof}
Three copyable vectors of $\delta$ form an orthogonal basis for $\C^3$ \cite{BCDPJV}.
Hence, we can get an invertible matrix $L$ that converts $\qbit{0}$, $\qbit{1}$, and $\qbit{2}$ into copyable vectors. This $L$ satisfies
\begin{equation}
\qbit{\psi} = (L \otimes L \otimes L)\qbit{\G}.
\end{equation}
\end{proof}

\begin{thm}
Let $\qbit{\psi}$ be a Frobenius state.
If there is $\qeff{\xi}$ such that $\qbit{\psi}$ with $\qeff{\xi}$ induces an ACFA, then $\qbit{\psi}$ is SLOCC-equivalent to $\qbit{\W}$.
\end{thm}
\begin{proof}
$\qbit{\psi}$ is a Frobenius state, so $\qbit{\psi}$ is SLOCC-equivalent to one of the states $\qbit{\G}$, $\qbit{\W}$, and $\qbit{\I}$.

First, assume $\qbit{\psi}$ be SLOCC-equivalent to $\qbit{\G}$.
According to Theorem \ref{slocccond}, there is $\qeff{\xi'}$ such that $\qbit{\G}$ with $\qeff{\xi'}$ induces an ACFA.
Let $\qeff{\xi'}:= \alpha\qeff{0} + \beta\qeff{1} + \gamma\qeff{2}$ with arbitrary complex numbers $\alpha$, $\beta$, and $\gamma$.
Due to the strong SLOCC-maximal condition (\ref{sloccmax}), $\alpha$, $\beta$, and $\gamma$ are restricted to nonzero.
A CFA induced by $\qbit{\G}$ with $\qeff{\xi'}$ does not satisfy the ACFA condition.
This contradicts the assumption.

By the same argument, a CFA induced by $\qbit{\I}$ with $\qeff{\xi'}$ is not an ACFA.
Therefore, $\qbit{\psi}$ is SLOCC-equivalent to $\qbit{\W}$.
\end{proof}

\begin{thm}
Let $\qbit{\psi}$ be a Frobenius state.
If there is $\qeff{\xi}$ such that $\qbit{\psi}$ with $\qeff{\xi}$ induces an ISCFA, then $\qbit{\psi}$ is SLOCC-equivalent to $\qbit{\I}$.
\end{thm}
\begin{proof}
This can be proved in the same way as the previous theorem.
\end{proof}

\subsection{Multiple Commutative Frobenius Algebras}
In \cite{BCAK}, it was shown that some pairs of two-dimensional SCFAs and ACFAs, which the authors called GHZ/W-pairs, can represent $N$-partite qubits. Here, we show that an SCFA, an ACFA, and an ISCFA can represent any qutrit.

We define a CFA trio $\T:= (\G, \W, \I)$, and the following arrows.
\begin{dfn}[Tick, Knurl, Wave]\rm
\begin{eqnarray}
\gph{-3}{4}{4}{0.3}{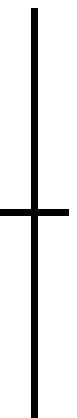} &:=& \gph{-3}{4}{4}{0.3}{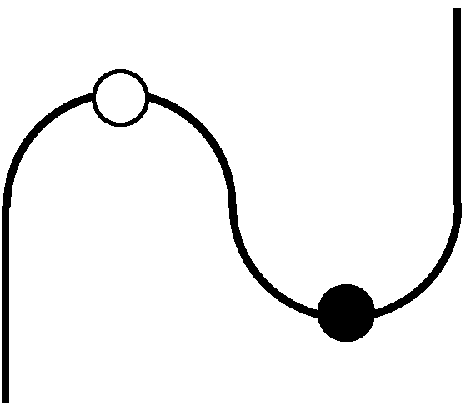}\\
\gph{-3}{4}{4}{0.3}{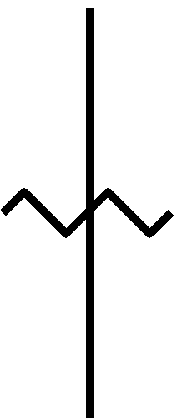} &:=& \gph{-3}{4}{4}{0.3}{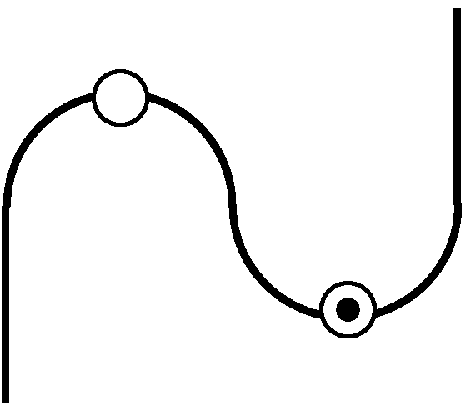}\\
\gph{-3}{4}{4}{0.3}{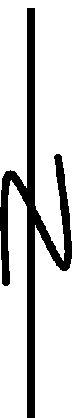} &:=& \gph{-3}{4}{4}{0.3}{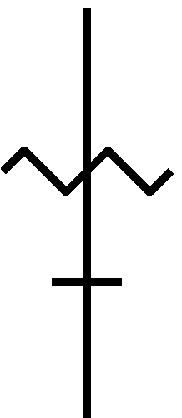}
\end{eqnarray}
\end{dfn}

First, we show how $\T$ can compose functions with the assistance of single vectors.
\begin{thm}\label{homcomp}
Any invertible $3 \times 3$ matrix $F$ can be represented by $\T$ and single qutrits.
\end{thm}
\begin{proof}
$F$ can be represented by an LDU decomposition as
\begin{equation}
F = PLDUP'
\end{equation}
$P$ and $P'$ are permutations, $L$ is a lower triangle matrix, $U$ is an upper triangle matrix, and $D$ is a diagonal matrix. Moreover, the diagonal elements of $L$ and $U$ are all $1$. In other words, $L$, $U$, and $D$ are
\begin{equation}
L = \mat{1}{0}{0}{l_2}{1}{0}{l_1}{l_0}{1}, \hspace{20pt}
D = \mat{d_0}{0}{0}{0}{d_1}{0}{0}{0}{d_2}, \hspace{20pt}
U = \mat{1}{u_0}{u_1}{0}{1}{u_2}{0}{0}{1}
\end{equation}
We define $\qbit{\psi}$, $\qbit{\phi}$, $\qbit{\pi}$, $\qbit{\eta}$, and $\qbit{\zeta}$ as
\begin{eqnarray}
\qbit{\psi} &:=& d_0\qbit{0}+d_1\qbit{1}+d_2\qbit{2}\\
\qbit{\phi} &:=& (u_2-u_0)\qbit{0}+\qbit{1}+\qbit{2}\\
\qbit{\pi} &:=& u_1\qbit{0}+u_0\qbit{1}+\qbit{2}\\
\qbit{\eta} &:=& (l_2-l_0)\qbit{0}+\qbit{1}+\qbit{2}\\
\qbit{\zeta} &:=& l_1\qbit{0}+l_0\qbit{1}+\qbit{2}
\end{eqnarray}
Using these qutrits, $L$, $D$, and $U$ are represented as
\begin{equation}
L = \gph{-8}{8}{8}{0.3}{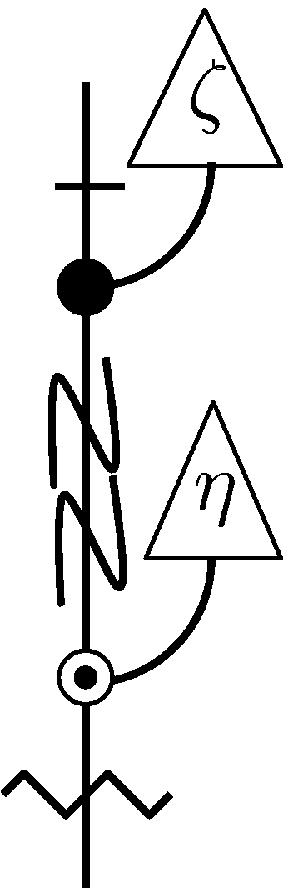}
\hspace{9ex}
D = \gph{-3}{4}{4}{0.3}{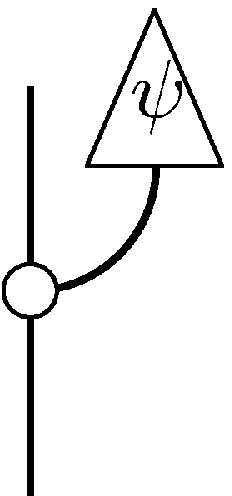}
\hspace{9ex}
U = \gph{-8}{8}{8}{0.3}{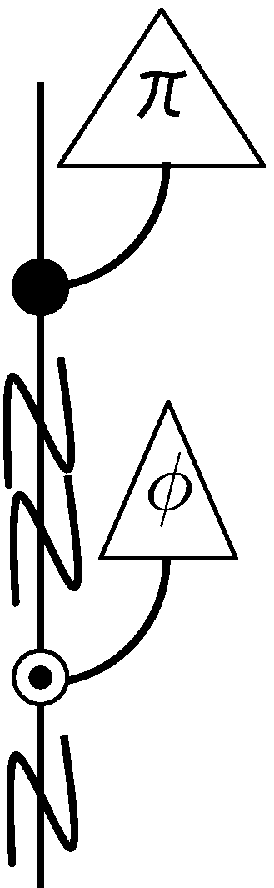}
\end{equation}
Any permutation can be composed of $\gph{-1}{2}{0}{0.18}{tick}$, $\gph{-1}{2}{0}{0.18}{knurl}$, and $\gph{-1}{2}{0}{0.18}{wave}$.
\end{proof}
Therefore, if a qutrit can be composed of $\T$ and single vectors, then all qutrits that are SLOCC-equivalent to another qutrit are composed of $\T$ and single vectors.

Moreover, $\T$ can `bring up' qutrits. In \cite{BCAK}, the authors defined a `quantum multiplexer' of qubits, which they called a QMUX. Here we define the qutrit version of QMUX.
\begin{dfn}[QMUX]\rm
A QMUX is 
\begin{eqnarray}
\gph{-4}{5}{5}{0.3}{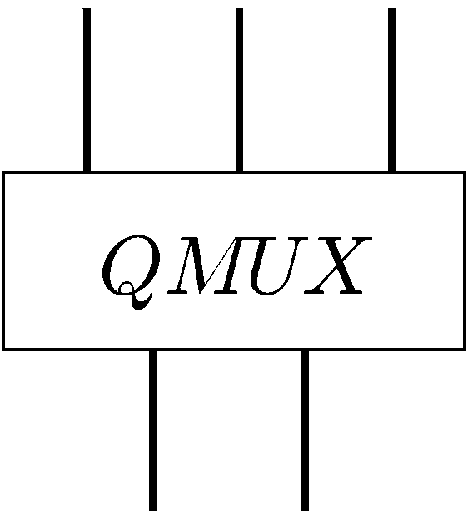}
:=\gph{-13}{13}{12}{0.3}{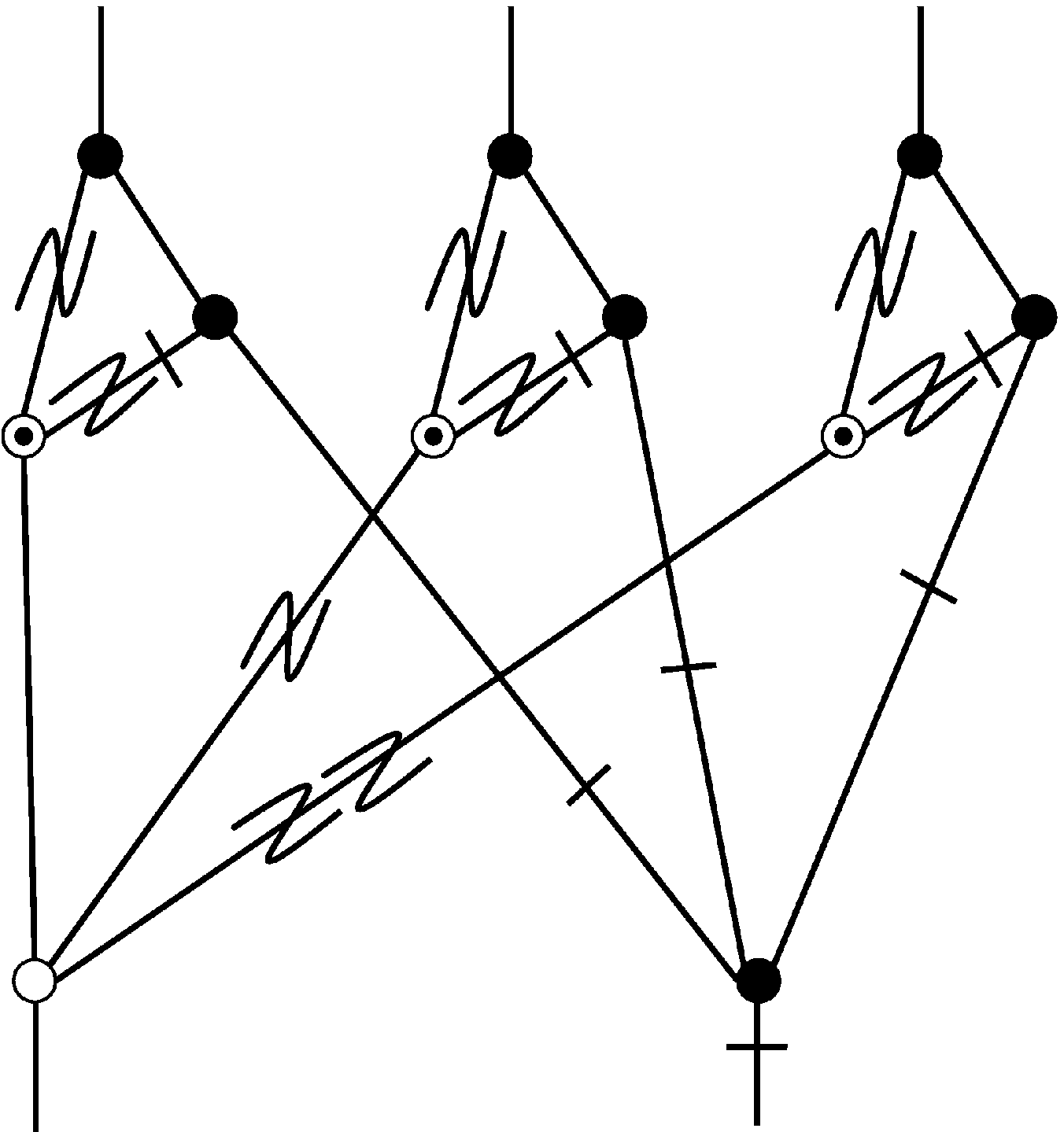}
\end{eqnarray}
\end{dfn}
\begin{thm}
$\T$ can compose a function that converts single qutrits $\qbit{\psi}\otimes\qbit{\phi}\otimes\qbit{\zeta}$ into $\qbit{0\psi}+\qbit{1\phi}+\qbit{2\zeta}$.
\end{thm}
\begin{proof}
Tracing each line, we can verify that QMUX converts single qutrits $\qbit{\psi}\otimes\qbit{\phi}\otimes\qbit{\zeta}$ into \[\qin{\phi}{2}\qin{\zeta}{2}\qbit{0\psi}+\qin{\zeta}{2}\qin{\psi}{2}\qbit{1\phi}+\qin{\psi}{2}\qin{\phi}{2}\qbit{2\zeta}.\] Based on Theorem \ref{homcomp}, there is an invertible matrix $L$ such that 
\begin{equation}
L = \mat{\frac{1}{\qin{\phi}{2}\qin{\zeta}{2}}}{0}{0}{0}{\frac{1}{\qin{\zeta}{2}\qin{\psi}{2}}}{0}{0}{0}{\frac{1}{\qin{\psi}{2}\qin{\phi}{2}}}
\end{equation}
\begin{equation}
\gph{-5}{5}{5}{0.3}{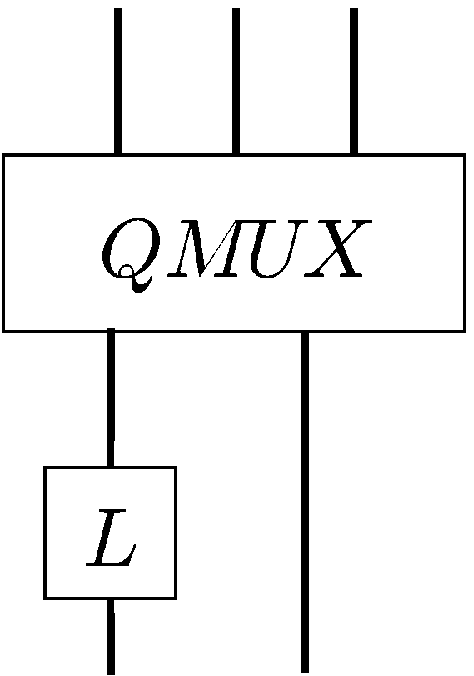}
\end{equation}
converts $\qbit{\psi}\otimes\qbit{\phi}\otimes\qbit{\zeta}$ into $\qbit{0\psi}+\qbit{1\phi}+\qbit{2\zeta}$.
\end{proof}
 Furthermore, by aligning QMUXs, this can be extended to the case of $N$-partite qutrits.
\begin{thm}
Let $\qbit{\psi}$, $\qbit{\phi}$, and $\qbit{\zeta}$ be $N$-partite qutrits. Then
\begin{eqnarray}
\gph{-17}{17}{17}{0.3}{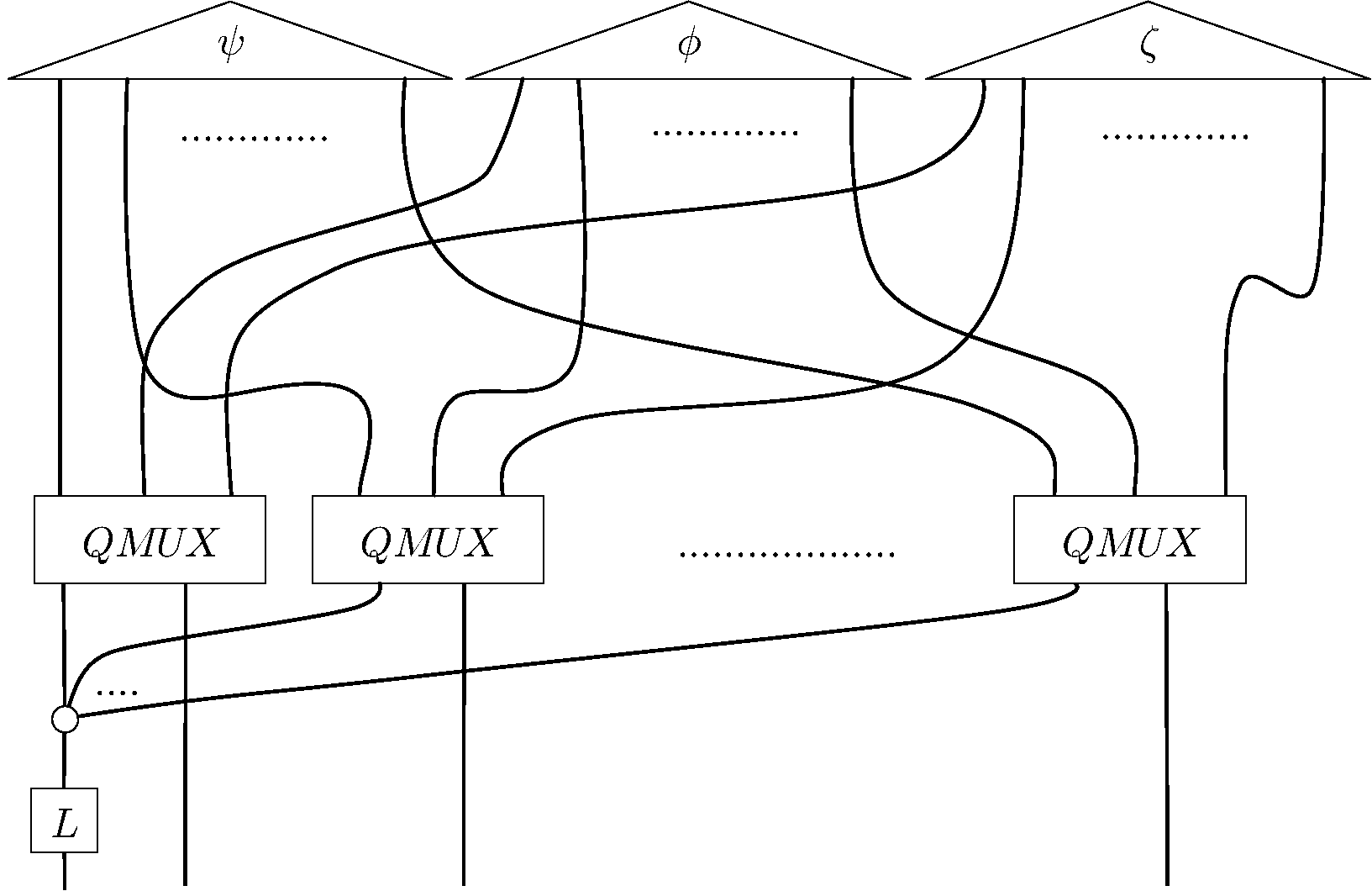}
\end{eqnarray}
is $\qbit{0\psi}+\qbit{1\phi}+\qbit{2\zeta}$ with some invertible matrix $L$.
\end{thm} 
Therefore, if any $N$-partite qutrit is composed of $\T$ and single vectors, then any N+1-partite qutrit is also composed of $\T$ and single vectors.

%-------------------
% Section 5.
%-------------------
\section{Related Work}

Abramsky and Coecke provided the graphical language for the categorical axioms of quantum protocols \cite{paper72}; see also \cite{Selinger2007139}.
That paper explicitly showed the graphical representation of quantum entanglement.
This was based on an interpretation of entanglement as a name and coname in a dagger compact closed category provided in \cite{paper72}, which used this interpretation to graphically express the flow of quantum information.
However, this representation was limited in that it was uniform, so it did not clarify the characteristics of the entanglement.

The paper \cite{BCAK} provided a graphical representation of the entangledness of a qubit. Using that representation, the authors classified SLOCC classes of entangled tripartite qubits.
In this paper, we use this method to express tripartite qutrits.

%-------------------
% Section 6.
%-------------------
\section{Conclusion}

We identified Frobenius classes in tripartite qutrits.
Frobenius states require strong SLOCC-maximality and strong symmetry.
Hence, by determining which classes were not strongly SLOCC-maximal, and which infinite SLOCC classes did not have a symmetric state, we were able to obtain three Frobenius classes that were SLOCC-equivalent to $\qbit{\G}$, $\qbit{\W}$, and $\qbit{\I}$. 
Then we  classified them further. One of these corresponded to an SCFA and one to an ACFA that corresponded to tripartite qubits on $\C^2$.
The other one was an ISCFA. We also proved that their correspondences were unique. The classification used the rank of $\gph{-1.4}{2}{1}{0.2}{dual}$.
Based on the equations for SCFA and ACFA, the ranks of $\gph{-1.4}{2}{1}{0.2}{dual}$ are $\dim{\Hilbert}$ and $1$, respectively.
Our definition of an ISCFA requires the rank of $\gph{-1.4}{2}{1}{0.2}{dual}$ to be neither $\dim{\Hilbert}$ nor $1$.
The uniqueness implies the algebraic and graphical structure of Frobenius states.

Finally, we desmonstrated the utility of the three CFAs. 
They can grow $N$-partite qutrits and construct any linear function using single qutrits. Furthermore, any multipartite qutrits can be expressed graphically using a graphical representation of the CFAs on $\C^3$.

Our method for expressing qutrits graphically is an extension of that used for qubits introduced in \cite{BCAK}.
The two methods have both similarities and differences. For example, both use Frobenius states that are highly symmetric and highly entangled, as well as SCFAs and ACFAs. Additionally, the CFAs used in both methods have the ability to make $N$-partite become same-dimensional systems with the help of single systems. This implies that these CFAs have some degree of completeness. In contrast, the differences include the fact that, although four SLOCC classes are not Frobenius classes in tripartite qubits, infinite SLOCC classes are not Frobenius classes in tripartite qutrits. Moreover, all non-symmetric SLOCC classes are non-maximal classes in tripartite qubits, whereas infinite non-symmetric SLOCC classes are not non-maximal classes in tripartite qutrits. Additionally, in tripartite qutrits, there is an ISCFA that is neither an SCFA nor an ACFA. Both characteristics are caused by the higher dimension of qutrits compared to qubits. Qubits are two-dimensional, which is the lowest possible dimension of an entangled state. Because of this low dimension, qubits do not have any non-symmetric classes. The second characteristic is also caused by the rank of $\gph{-1.4}{2}{1}{0.2}{dual}$.
In qubits, the rank of a nontrivial function is limited to $1$ or $\dim{\Hilbert} = 2$.
However, in qutrits, there is an intermediate rank, i.e., $2$.
Therefore, there is an ISCFA that does not exist in qubits.

We demonstrated the correspondence between CFAs and some tripartite qutrits.
However, an infinite number of SLOCC classes are possible that do not have graphical representations reflecting their entanglement properties.
The success of Frobenius states implies the algebraic structure of tripartite qutrits. 
Another algebra is needed to express other SLOCC classes graphically to specify their features. Many classes are not symmetric; in such cases, commutative properties are not needed for the algebras.

In higher-dimensional tripartite systems, Frobenius classes may exist.
Considering the rank of $\gph{-1.4}{2}{1}{0.2}{dual}$, more classifications may be possible for higher dimensions.
However, the number of these classifications may remain finite, even though SLOCC classes are infinite.
Furthermore, an ISCFA may or may not exist at higher dimensions. If it does, research into the normal forms of CFAs with ticks and knurls such as those described in \cite{Roy2010} would be helpful.

\section*{Acknowledgments}
I would like to express my deep gratitude to Prof.~Masami Hagiya and Dr.~Yoshihiko Kakutani for their guidance and advice. I also would like to thank the referees for their comments.

\nocite{Quantum}
%-------------------
\bibliographystyle{eptcs}
\bibliography{bibfile}

\begin{thebibliography}{10}
\providecommand{\bibitemdeclare}[2]{}
\providecommand{\urlprefix}{Available at }
\providecommand{\url}[1]{\texttt{#1}}
\providecommand{\href}[2]{\texttt{#2}}
\providecommand{\urlalt}[2]{\href{#1}{#2}}
\providecommand{\doi}[1]{doi:\urlalt{http://dx.doi.org/#1}{#1}}
\providecommand{\bibinfo}[2]{#2}

\bibitemdeclare{inproceedings}{paper72}
\bibitem{paper72}
\bibinfo{author}{S.~Abramsky} \& \bibinfo{author}{B.~Coecke}
  (\bibinfo{year}{2004}): \emph{\bibinfo{title}{A Categorical Semantics of
  Quantum Protocols}}.
\newblock In: {\sl \bibinfo{booktitle}{Proceedings of the 19th Annual {IEEE}
  Symposium on Logic in Computer Science: {LICS} 2004}},
  \bibinfo{publisher}{{IEEE} Computer Society}, pp. \bibinfo{pages}{415--425},
  \doi{10.1109/LICS.2004.1319636}.

\bibitemdeclare{article}{PhysRevLett.85.3313}
\bibitem{PhysRevLett.85.3313}
\bibinfo{author}{Helle Bechmann-Pasquinucci} \& \bibinfo{author}{Asher Peres}
  (\bibinfo{year}{2000}): \emph{\bibinfo{title}{Quantum Cryptography with
  3-State Systems}}.
\newblock {\sl \bibinfo{journal}{Phys. Rev. Lett.}}
  \bibinfo{volume}{85}(\bibinfo{number}{15}), pp. \bibinfo{pages}{3313--3316},
  \doi{10.1103/PhysRevLett.85.3313}.

\bibitemdeclare{article}{PhysRevLett.70.1895}
\bibitem{PhysRevLett.70.1895}
\bibinfo{author}{Charles~H. Bennett}, \bibinfo{author}{Gilles Brassard},
  \bibinfo{author}{Claude Cr\'epeau}, \bibinfo{author}{Richard Jozsa},
  \bibinfo{author}{Asher Peres} \& \bibinfo{author}{William~K. Wootters}
  (\bibinfo{year}{1993}): \emph{\bibinfo{title}{Teleporting an unknown quantum
  state via dual classical and Einstein-Podolsky-Rosen channels}}.
\newblock {\sl \bibinfo{journal}{Phys. Rev. Lett.}}
  \bibinfo{volume}{70}(\bibinfo{number}{13}), pp. \bibinfo{pages}{1895--1899},
  \doi{10.1103/PhysRevLett.70.1895}.

\bibitemdeclare{article}{PhysRevA.64.012306}
\bibitem{PhysRevA.64.012306}
\bibinfo{author}{Mohamed Bourennane}, \bibinfo{author}{Anders Karlsson} \&
  \bibinfo{author}{Gunnar Bj\"ork} (\bibinfo{year}{2001}):
  \emph{\bibinfo{title}{Quantum key distribution using multilevel encoding}}.
\newblock {\sl \bibinfo{journal}{Phys. Rev. A}}
  \bibinfo{volume}{64}(\bibinfo{number}{1}), p. \bibinfo{pages}{012306},
  \doi{10.1103/PhysRevA.64.012306}.

\bibitemdeclare{article}{PhysRevLett.88.127901}
\bibitem{PhysRevLett.88.127901}
\bibinfo{author}{D.~Bru\ss{}} \& \bibinfo{author}{C.~Macchiavello}
  (\bibinfo{year}{2002}): \emph{\bibinfo{title}{Optimal Eavesdropping in
  Cryptography with Three-Dimensional Quantum States}}.
\newblock {\sl \bibinfo{journal}{Phys. Rev. Lett.}}
  \bibinfo{volume}{88}(\bibinfo{number}{12}), p. \bibinfo{pages}{127901},
  \doi{10.1103/PhysRevLett.88.127901}.

\bibitemdeclare{conference}{BCAK}
\bibitem{BCAK}
\bibinfo{author}{B.~Coecke} \& \bibinfo{author}{A.~Kissinger}
  (\bibinfo{year}{2010}): \emph{\bibinfo{title}{The compositional structure of
  multipartite quantum entanglement}}.
\newblock In: {\sl \bibinfo{booktitle}{Automata, Languages and Programming}},
  \bibinfo{volume}{37th International Colloquium, ICALP 2010, Bordeaux, France,
  July 6-10, 2010, Proceedings, Part II}, \doi{10.1007/978-3-642-14162-1\_25}.

\bibitemdeclare{article}{BCDPJV}
\bibitem{BCDPJV}
\bibinfo{author}{B.~Coecke}, \bibinfo{author}{D.~Pavlovic} \&
  \bibinfo{author}{J.~Vicary} (\bibinfo{year}{2009}): \emph{\bibinfo{title}{A
  new description of orthogonal bases}}.
\newblock {\sl \bibinfo{journal}{Mathematical Structures in Computer Science}}
  \bibinfo{note}{13 pp., to appear, arXiv:0810.0812}.

\bibitemdeclare{article}{PhysRevA.62.062314}
\bibitem{PhysRevA.62.062314}
\bibinfo{author}{W.~D\"ur}, \bibinfo{author}{G.~Vidal} \&
  \bibinfo{author}{J.~I. Cirac} (\bibinfo{year}{2000}):
  \emph{\bibinfo{title}{Three qubits can be entangled in two inequivalent
  ways}}.
\newblock {\sl \bibinfo{journal}{Phys. Rev. A}}
  \bibinfo{volume}{62}(\bibinfo{number}{6}), p. \bibinfo{pages}{062314},
  \doi{10.1103/PhysRevA.62.062314}.

\bibitemdeclare{article}{PhysRevA.74.052336}
\bibitem{PhysRevA.74.052336}
\bibinfo{author}{L.~Lamata}, \bibinfo{author}{J.~Le\'on},
  \bibinfo{author}{D.~Salgado} \& \bibinfo{author}{E.~Solano}
  (\bibinfo{year}{2006}): \emph{\bibinfo{title}{Inductive classification of
  multipartite entanglement under stochastic local operations and classical
  communication}}.
\newblock {\sl \bibinfo{journal}{Phys. Rev. A}}
  \bibinfo{volume}{74}(\bibinfo{number}{5}), p. \bibinfo{pages}{052336},
  \doi{10.1103/PhysRevA.74.052336}.

\bibitemdeclare{article}{PhysRevA.75.022318}
\bibitem{PhysRevA.75.022318}
\bibinfo{author}{L.~Lamata}, \bibinfo{author}{J.~Le\'on},
  \bibinfo{author}{D.~Salgado} \& \bibinfo{author}{E.~Solano}
  (\bibinfo{year}{2007}): \emph{\bibinfo{title}{Inductive entanglement
  classification of four qubits under stochastic local operations and classical
  communication}}.
\newblock {\sl \bibinfo{journal}{Phys. Rev. A}}
  \bibinfo{volume}{75}(\bibinfo{number}{2}), p. \bibinfo{pages}{022318},
  \doi{10.1103/PhysRevA.75.022318}.

\bibitemdeclare{book}{MacLane}
\bibitem{MacLane}
\bibinfo{author}{Saunders~Mac Lane} (\bibinfo{year}{1998}):
  \emph{\bibinfo{title}{{Categories for the Working Mathematician (Graduate
  Texts in Mathematics)}}}, \bibinfo{edition}{2nd} edition.
\newblock \bibinfo{publisher}{Springer}.

\bibitemdeclare{article}{PhysRevA.81.052315}
\bibitem{PhysRevA.81.052315}
\bibinfo{author}{P.~Mathonet}, \bibinfo{author}{S.~Krins},
  \bibinfo{author}{M.~Godefroid}, \bibinfo{author}{L.~Lamata},
  \bibinfo{author}{E.~Solano} \& \bibinfo{author}{T.~Bastin}
  (\bibinfo{year}{2010}): \emph{\bibinfo{title}{Entanglement equivalence of
  $N$-qubit symmetric states}}.
\newblock {\sl \bibinfo{journal}{Phys. Rev. A}}
  \bibinfo{volume}{81}(\bibinfo{number}{5}), p. \bibinfo{pages}{052315},
  \doi{10.1103/PhysRevA.81.052315}.

\bibitemdeclare{book}{Quantum}
\bibitem{Quantum}
\bibinfo{author}{Michael~A. Nielsen} \& \bibinfo{author}{Isaac~L. Chuang}
  (\bibinfo{year}{2000}): \emph{\bibinfo{title}{{Quantum Computation and
  Quantum Information}}}.
\newblock \bibinfo{publisher}{Cambridge University Press},
  \doi{10.2277/0521635039}.

\bibitemdeclare{mastersthesis}{Roy2010}
\bibitem{Roy2010}
\bibinfo{author}{Shibdas Roy} (\bibinfo{year}{2010}): \emph{\bibinfo{title}{A
  Compositional Characterization of Multipartite Quantum States}}.
\newblock Master's thesis, \bibinfo{school}{University of Oxford}.
\newblock \urlprefix\url{http://www.cs.ox.ac.uk/people/bob.coecke/Shibdas.pdf}.

\bibitemdeclare{article}{Selinger04}
\bibitem{Selinger04}
\bibinfo{author}{Peter Selinger} (\bibinfo{year}{2004}):
  \emph{\bibinfo{title}{Towards a quantum programming language}}.
\newblock {\sl \bibinfo{journal}{Mathematical Structures in Computer Science}}
  \bibinfo{volume}{14}, pp. \bibinfo{pages}{527--586},
  \doi{10.1017/S0960129504004256}.

\bibitemdeclare{article}{Selinger2007139}
\bibitem{Selinger2007139}
\bibinfo{author}{Peter Selinger} (\bibinfo{year}{2007}):
  \emph{\bibinfo{title}{Dagger Compact Closed Categories and Completely
  Positive Maps: (Extended Abstract)}}.
\newblock {\sl \bibinfo{journal}{Electronic Notes in Theoretical Computer
  Science}} \bibinfo{volume}{170}, pp. \bibinfo{pages}{139 -- 163},
  \doi{10.1016/j.entcs.2006.12.018}.
\newblock \bibinfo{note}{Proceedings of the 3rd International Workshop on
  Quantum Programming Languages (QPL 2005)}.

\bibitemdeclare{incollection}{Selinger09Survey}
\bibitem{Selinger09Survey}
\bibinfo{author}{Peter Selinger} (\bibinfo{year}{2011}):
  \emph{\bibinfo{title}{A Survey of Graphical Languages for Monoidal
  Categories}}.
\newblock In \bibinfo{editor}{Bob Coecke}, editor: {\sl \bibinfo{booktitle}{New
  Structures for Physics}}, \bibinfo{edition}{1st} edition,
  chapter~\bibinfo{chapter}{4}, {\sl \bibinfo{series}{Lecture Notes in
  Physics}} \bibinfo{volume}{813}, \bibinfo{publisher}{Springer}, pp.
  \bibinfo{pages}{289--355}, \doi{10.1007/978-3-642-12821-9\_4}.

\bibitemdeclare{article}{PhysRevA.65.052112}
\bibitem{PhysRevA.65.052112}
\bibinfo{author}{F.~Verstraete}, \bibinfo{author}{J.~Dehaene},
  \bibinfo{author}{B.~De~Moor} \& \bibinfo{author}{H.~Verschelde}
  (\bibinfo{year}{2002}): \emph{\bibinfo{title}{Four qubits can be entangled in
  nine different ways}}.
\newblock {\sl \bibinfo{journal}{Phys. Rev. A}}
  \bibinfo{volume}{65}(\bibinfo{number}{5}), p. \bibinfo{pages}{052112},
  \doi{10.1103/PhysRevA.65.052112}.

\bibitemdeclare{article}{Wootters1982}
\bibitem{Wootters1982}
\bibinfo{author}{W.~K. Wootters} \& \bibinfo{author}{W.~H. Zurek}
  (\bibinfo{year}{1982}): \emph{\bibinfo{title}{A single quantum cannot be
  cloned}}.
\newblock {\sl \bibinfo{journal}{Nature}}
  \bibinfo{volume}{299}(\bibinfo{number}{5886}), pp. \bibinfo{pages}{802--803},
  \doi{10.1038/299802a0}.

\bibitemdeclare{article}{qutrit}
\bibitem{qutrit}
\bibinfo{author}{Xin-Gang Yang}, \bibinfo{author}{Zhi-Xi Wang},
  \bibinfo{author}{Xiao-Hong Wang} \& \bibinfo{author}{Shao-Ming Fei}
  (\bibinfo{year}{2008}): \emph{\bibinfo{title}{Classification of Bipartite and
  Tripartite Qutrit Entanglement under {SLOCC}}}.
\newblock {\sl \bibinfo{journal}{Communications in Theoretical Physics}}
  \bibinfo{volume}{50}(\bibinfo{number}{3}), pp. \bibinfo{pages}{651--654},
  \doi{10.1088/0253-6102/50/3/25}.

\end{thebibliography}
 %-------------------
\end{document}